\documentclass[twoside,leqno,twocolumn]{article}

\usepackage[letterpaper]{geometry}

\usepackage{siamproceedings}

\usepackage[T1]{fontenc}
\usepackage{amsfonts, amssymb}
\usepackage{microtype}
\usepackage{graphicx}
\usepackage{enumitem}
\usepackage{algorithm}
\usepackage[noend]{algorithmic}
\usepackage{thm-restate}
\usepackage{complexity}
\usepackage{siunitx}
\ifpdf
  \DeclareGraphicsExtensions{.eps,.pdf,.png,.jpg}
\else
  \DeclareGraphicsExtensions{.eps}
\fi

\newsiamremark{remark}{Remark}
\newsiamremark{hypothesis}{Hypothesis}
\newsiamremark{observation}{Observation}
\crefname{hypothesis}{Hypothesis}{Hypotheses}
\newsiamthm{claim}{Claim}

\usepackage{amsopn}

\let\UP\undefinedcommand
\DeclareMathOperator{\UP}{UP}
\newcommand{\tisp}{\ensuremath{t\textsc{-Isp}}}
\newcommand{\tispac}{\ensuremath{t\textsc{-Isp-Ac}}}
\newcommand{\tispot}{\ensuremath{t\textsc{-Isp-Ot}}}
\newcommand{\tispacot}{\ensuremath{t\textsc{-Isp-Ac-Ot}}}

\newcommand{\pisp}{\ensuremath{2\textsc{-Isp}}}
\newcommand{\old}[1]{{}}

\begin{document}

\newcommand\relatedversion{}
\renewcommand\relatedversion{\thanks{This is the full version of a paper of the same title that is to appear at ALENEX 2026.}} 

\title{\Large Efficient Heuristics and Exact Methods for Pairwise Interaction Sampling\relatedversion}
  \author{Sándor P. Fekete\thanks{TU Braunschweig, Algorithms Group (\email{s.fekete@tu-bs.de}, \email{keldenich@ibr.cs.tu-bs.de}, \email{krupke@ibr.cs.tu-bs.de}, \email{perk@ibr.cs.tu-bs.de}).}
  \and Phillip Keldenich\footnotemark[2]
  \and Dominik Krupke\footnotemark[2]
  \and Michael Perk\footnotemark[2]}
\date{}

\maketitle







\newcommand{\symfeaturemodel}[0]{\ensuremath{\varphi}}
\newcommand{\symsimplified}[0]{\ensuremath{\psi}}

\newcommand{\symqueue}[0]{\ensuremath{\mathcal{Q}}}

\newcommand{\symworkingsample}[0]{\ensuremath{\mathcal{S}}}

\newcommand{\symsample}[0]{\ensuremath{\mathcal{S}}}

\newcommand{\symfeasibleinteractions}[0]{\ensuremath{\mathcal{I}}}

\newcommand{\syminteraction}[0]{\ensuremath{I}}

\newcommand{\symtrail}[0]{\ensuremath{T}}

\newcommand{\sympartialconfig}[0]{\ensuremath{Q}}

\newcommand{\symconfiguration}[0]{\ensuremath{C}}

\newcommand{\symtargetqueuesize}[0]{\ensuremath{k}}

\newcommand{\symfeatureset}[0]{\ensuremath{\mathcal{F}}}
\newcommand{\symconcretefeatureset}[0]{\ensuremath{\mathcal{C}}}

\newcommand{\symdestroyparameters}[0]{\ensuremath{P_d}}
\newcommand{\symdestruction}[0]{\ensuremath{\mathcal{R}}}

\newcommand{\symrepairparameters}[0]{\ensuremath{P_r}}

\newcommand{\symtargetrepairsize}[0]{\ensuremath{s}}

\newcommand{\symtargetsamplesize}[0]{\ensuremath{s}}

\newcommand{\symsatoracle}[0]{\ensuremath{\mathcal{A}}}

\newcommand{\sympispinstance}[0]{\ensuremath{\mathcal{B}}}

\newcommand{\symredextraconcrete}[0]{\ensuremath{\xi}}

\newcommand{\symredformula}[0]{\ensuremath{\vartheta}}

\newcommand{\symconccomponent}[0]{\ensuremath{K}}

\newcommand{\symclause}[0]{\ensuremath{\gamma}}

\newcommand{\symnumthreads}[0]{\ensuremath{\rho}}

\newcommand{\symindset}[0]{\ensuremath{D}}
\newcommand{\symindsetfamily}[0]{\ensuremath{\mathcal{D}}}

\newcommand{\symmaxexclusiveset}[0]{\ensuremath{\mathcal{U}}}

\newcommand{\symsat}[0]{\textsc{Sat}}

\newcommand{\symunsat}[0]{\textsc{Unsat}}

\begin{abstract} 
  We consider a class of optimization problems that are fundamental to testing in
modern configurable software systems, e.g., in automotive industries. In
\emph{pairwise interaction sampling}, we are given a (potentially very large)
configuration space, in which each dimension corresponds to a possible Boolean
feature of a software system; valid configurations are the satisfying
assignments of a given propositional formula $\symfeaturemodel{}$. The objective is to
find a minimum-sized family of configurations, such that each pair of features is
jointly tested at least once. Due to its relevance in Software Engineering,
this problem has been studied extensively for over 20 years.

In addition to (1) new theoretical insights (we prove \BH-hardness), we provide 
a broad spectrum of key contributions on the practical side that allow substantial
progress for the practical performance.  These include the following. (2) We devise and engineer an initial
solution algorithm that can find solutions and correctly identify the set of
valid interactions in reasonable time, even for very large instances with
hundreds of millions of valid interactions, which previous approaches could not solve in such time.
(3) We present an enhanced approach for computing lower bounds.
(4) We present an ex\-act algorithm to find optimal solutions based on an interaction
selection heuristic driving an incremental SAT solver
that works on small and medium-sized instances.
(5) For larger instances, we present a meta\-heuristic solver based
on large neighborhood search (LNS) that solves most of the
instances in a diverse benchmark library of instances 
to provable optimality within an hour on commodity hardware.
(6) Remarkably, we are able to solve the largest instances we found in 
published benchmark sets (with about \num{500000000} feasible interactions)
to provable optimality.
Previous approaches were not even able to compute feasible solutions.

\end{abstract}

\section{Introduction.}
Many modern software systems are configurable, sometimes with thousands of options that are often interdependent~\cite{HZS+:EMSE16,QPV+:SoSyM17,sundermann2023evaluating}; examples for complex, highly configurable systems include automobiles as well as operating system kernels.
Often, bugs only manifest if a certain combination of features is selected~\cite{BDC+:SETSS89,CKMR03,ABKS13,AMS+:TOSEM18}.
A popular solution for fault detection in such systems is product-based testing combined with sampling
methods to create a representative list of configurations to be tested individually~\cite{RAK+:VaMoS13,MMCA:IST14,VAT+:SPLC18};
a popular way of ensuring that a sample is at least somewhat representative 
is to require coverage of all interactions between at most $t$ features~\cite{OGB:SPLC19,FVDF:JSS21} for some constant $t$.

In 
this paper, a configurable software system is modeled as a propositional formula $\symfeaturemodel{}(x_1,\ldots,x_n)$ on a set $\symfeatureset{} = \{x_1,\ldots,x_n\}$ of Boolean variables, also called \emph{features}.
A \emph{configuration} is an assignment of truth values to the features $x_1, \ldots, x_n$; it is \emph{valid} if it satisfies $\symfeaturemodel{}$.
Each feature $x_i$ has two possible \emph{literals}:
$x_i$ (positive) and $\overline{x_i}$ (negative).
Thus, a configuration can also be represented as a set of feature literals.
In the following, we use the term \emph{configuration} to refer only to valid configurations.

We assume $\symfeaturemodel{}$ to be given in conjunctive normal form (CNF), i.e., as a conjunction of disjunctive clauses, each consisting of feature literals.
We consider the clauses $\symclause{}_j$ of $\symfeaturemodel{}$ to be sets of literals.
By $\symconcretefeatureset{} \subseteq \symfeatureset{}$, we denote a subset of features called \emph{concrete features}.
For a concrete feature $x$, the literals $x$ and $\overline{x}$ are \emph{concrete literals}.
A set $\syminteraction{}$ of $|\syminteraction{}| = t$ concrete literals is called an \emph{interaction} of size $t$ or \emph{$t$-wise interaction}.

An interaction is \emph{feasible} if there is a valid configuration $\symconfiguration{} \supseteq \syminteraction{}$; in that case, we say that $\symconfiguration{}$ covers $\syminteraction{}$.
The $t$-wise interaction sampling problem (\tisp) asks, for a given $\symfeaturemodel{}$ and $\symconcretefeatureset{}$, for a minimum cardinality set $\symsample{}$ of valid configurations that achieves what we call \emph{$t$-wise coverage}.
A set $\symsample{}$ of configurations, also called \emph{sample}, is said to have \emph{$t$-wise coverage} if every feasible interaction $\syminteraction{}$ is covered by a configuration $\symconfiguration{} \in \symsample{}$.
As an example, consider the model $\symfeaturemodel{}(x,y,z) = x \vee y$.
Except for $\{\overline{x}, \overline{y}\}$, all pairwise interactions are feasible.
A minimum-cardinality sample with pairwise coverage for this example is $\symsample{} = \{
  \{x, y, z\},
  \{x, \overline{y}, \overline{z}\},
  \{x, \overline{y}, z\},
  \{\overline{x}, y, \overline{z}\},
  \{\overline{x}, y, z\}
\}$.

Following relevant previous work, we mostly treat the case $t = 2$,
i.e., we are searching for a minimum-cardinality sample with \emph{pairwise coverage}.
Our implementation currently only handles the case $t = 2$, although the overall approach could be scaled to larger $t$ as well.
This would, however, require a number of changes to the implementation and some re-engineering to improve scalability.
In addition, the number of interactions tends to grow drastically with higher $t$,
so any approach that eventually relies on enumerating the set of feasible interactions will not scale to the largest instances for larger $t$.
In most applications that we are aware of, $t \in [1, 3]$; this is most likely due to that growth as well.

We make the following contributions.


\subsection*{Complexity.} We consider the complexity of the problem, showing that it can be solved in polynomial time 
  with logarithmically many queries to a \textsc{Sat} oracle, but not with constantly many such queries, unless the polynomial time hierarchy collapses.

\subsection*{Initial Heuristic.} We devise and engineer an initial solution algorithm that can find solutions 
                         and correctly identify the set of valid interactions in reasonable time, 
                         even for very large instances with hundreds of millions of valid interactions,
                         which previous approaches could not solve in such time.

\subsection*{Lower Bounds.} We present a cut, price \& round approach for computing lower bounds on the size of an optimal sample.
                    Lower bounds are crucial to any approach that hopes to identify provably minimal samples with pairwise coverage.
                    Aside from their immediate use to establish optimality, we also use lower bounds to break symmetries and to provide improved starting points for heuristics.

\subsection*{Exact Algorithm.} We present an exact algorithm to find optimal solutions based on an interaction selection heuristic driving an incremental SAT solver that works on small and medium-sized instances.

\subsection*{Metaheuristic Solver.} We combine our lower bounds, heuristic insights and several exact solution approaches 
                             and engineer a metaheuristic solver based on large neighborhood search (LNS).
                             This solver runs a portfolio of different approaches in parallel to search for better solutions and lower bounds,
                             solving \qty{85}{\percent} of instances in a diverse benchmark of instances to provable optimality within an hour,
                             significantly outperforming previous approaches.

\subsection*{Preprocessing Techniques.} We adapt several well-known SAT preprocessing techniques to our problem and evaluate their impact on the solution process.

\subsection*{Progress on Benchmarks.} Aside from improved bounds and optimal solutions for many publicly available benchmarks,
                               we solved the four largest publicly available instances 
                               of the PSPL Scalability Challenge~\cite{PTR+:SPLC19} to provable optimality.
                               Previous approaches could not even find feasible solutions for these instances.

\subsection*{Paper Structure.} The remainder of the paper is structured as follows.
\Cref{sec:related-work} discusses related work, 
\cref{sec:preliminaries} defines some more formal terms, and \cref{sec:complexity} presents our complexity results.
\Cref{sec:initial-heuristic,sec:lower-bounds,sec:sammy} describe the individual parts of our algorithm and implementation; see \cref{fig:solver_overview} for an outline.
In \cref{sec:experiments}, we present our experimental evaluation; \cref{sec:conclusion} concludes.
\begin{figure}
\centering
\includegraphics[width=.98\linewidth]{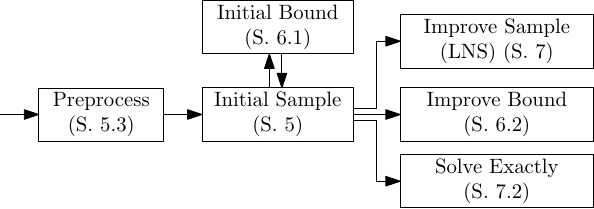}
\caption{An overview of the flow of our algorithm and the corresponding sections in the paper.
         The three rightmost boxes are executed in parallel,
         with multiple copies of the LNS using different strategies running simultaneously.}
\label{fig:solver_overview}
\end{figure}

\section{Related Work.}
\label{sec:related-work}
There are various problem variants, some of which can be transformed into one another.
The primary differences lie in the arities of the features (binary, $g$-ary, or mixed) and the types of constraints (unconstrained, forbidden pairs of feature values, or arbitrary propositional formulas).
Although we focus on binary features, our model's ability to specify concrete features and support arbitrary propositional formulas enables it to represent all of these variants.
Most of the literature presented below also supports $t$-wise coverage for $t \geq 2$.

For binary features and forbidden pairs, the configuration space reduces to a 2-SAT problem, which can be solved in linear time~\cite{aspvall1979linear}.
However, Maltais and Moura~\cite{maltais2010finding} demonstrated that minimizing the sample size remains not only \NP-hard but also hard to approximate; the optimal solution corresponds to a minimum edge clique cover.
As this problem is a special case of our formulation (restricted to concrete features and a 2-SAT formula) this hardness result extends to our variant.

There are multiple exact algorithms based on SAT solving.  
Nanba et al.~\cite{nanba2012using} employ a binary search strategy to find optimal solutions, successfully solving instances with up to 80 features.  
The CALOT algorithm~\cite{yamada2015optimization} uses incremental SAT solving to iteratively reduce the number of available configurations until infeasibility is detected, thereby certifying the previous solution as optimal.  
Ans{\'o}tegui et al.~\cite{ansotegui2022incomplete} delegate the minimization to a MaxSAT solver and are able to compute optimal solutions for several non-trivial instances.  
These approaches also incorporate simple lower-bounding techniques for symmetry breaking, which are effective primarily for high arity $g$.  
Yang et al.~\cite{yang2021method} propose a more advanced lower-bounding method based on decomposition along a prohibited edge.  
Hervieu~\cite{hervieu2016practical} employs constraint programming instead of a SAT-solver.
Recently, Krupke~et~al.~\cite{10.1145/3712193} introduced a strategy for computing provable lower bounds via a suitable dual problem.
They were the first to implement a large neighborhood search (LNS) and demonstrated its effectiveness in producing optimal or near-optimal solutions.  
While we use different methods to compute lower bounds and solve subproblems, their work provides the foundation for our LNS approach.

There are various heuristics in the literature.
ACTS~\cite{yu2013acts} provides several IPOG-based~\cite{LK+:ECBS07} heuristics that construct multiple configurations in parallel.
The variant of Duan et al.~\cite{duan2017optimizing} models the problem as a hypergraph coloring problem, where interactions are vertices and conflicts are encoded as hyperedges; they then use the degree of conflict to guide a greedy coloring strategy.
YASA~\cite{KTS+:VaMoS20}, which also builds on IPOG, makes extensive use of a \symsat{} solver to ensure that each partial configuration can be completed to a valid one.
It further applies local optimization steps at the end to reduce the sample size.
Other approaches, such as IncLing~\cite{AKT+:GPCE16}, generate one configuration at a time, prioritizing underrepresented feature interactions by leveraging feature frequency and polarity in the remaining uncovered pairs.
Yamada et al.~\cite{yamada2016greedy} also generate configurations individually, but integrate a CDCL SAT solver more deeply than other methods.
Their approach constructs configurations during the solver's decision-making process, with possible amendments if conflicts arise.
Ans{\'o}tegui et al.~\cite{ansotegui_et_al:LIPIcs.CP.2021.12} likewise build one configuration at a time using a MaxSAT solver, which can also modify recently added configurations within a dynamic window.
Kadioglu~\cite{kadioglu2017column} proposed a column generation heuristic, but evaluated it only on small, unconstrained instances.
CAmpactor~\cite{zhao2023campactor} iteratively removes one configuration and attempts to restore full coverage by iteratively selecting a missing interaction and modifying another configuration to include it.

\section{Preliminaries.}
\label{sec:preliminaries}
Let $\sympartialconfig{}$ be a \emph{partial configuration} of $\symfeaturemodel{}$, i.e., a set of literals such that $\{\ell, \overline{\ell}\} \nsubseteq \sympartialconfig{}$ for any $\ell$.
We say that a literal is \emph{true} in $\sympartialconfig{}$ iff $\ell \in \sympartialconfig{}$, \emph{false} in $\sympartialconfig{}$ iff $\overline{\ell} \in \sympartialconfig{}$ and \emph{open} in $\sympartialconfig{}$ if neither is the case.
Let $\symclause{}_j$ be a clause of $\symfeaturemodel{}$; we say that $\sympartialconfig{}$ satisfies $\symclause{}_j$ if $\sympartialconfig{} \cap \symclause{}_j \neq \emptyset$, i.e., one of the literals of $\symclause{}_j$ is true in $\sympartialconfig{}$, and that $\sympartialconfig{}$ violates $\symclause{}_j$ if $\overline{\symclause{}_j} := \{\overline{\ell} \mid \ell \in \symclause{}_j\} \subseteq \sympartialconfig{}$, i.e., if all the literals of $\symclause{}_j$ are false in $\sympartialconfig{}$.
We say that a clause $\symclause{}_j$ is \emph{unit} under $\sympartialconfig{}$ if all but one of its literals are false in $\sympartialconfig{}$ and the remaining literal $\ell$ is open in $\sympartialconfig{}$.
\emph{Unit Propagation} (UP), also called \emph{Boolean Constraint Propagation}, extends a partial configuration $\sympartialconfig{}$ by adding the open literals $\ell$ of unit clauses to $\sympartialconfig{}$ until either of the following happens:
there are no more unit clauses, or there is a violated clause.
In the latter case, we say that UP encountered a conflict.
In the following, by $\UP(\sympartialconfig{})$, we denote the partial assignment resulting from applying UP to $\sympartialconfig{}$; if UP encounters a conflict, we write $\UP(\sympartialconfig{}) = \bot$.

\section{Complexity.}
\label{sec:complexity}
Now we analyze the complexity of the \tisp{},
in particular of the decision version in which we are given a formula $\symfeaturemodel{}$ on Boolean features $\symfeatureset{}$, 
a set of concrete features $\symconcretefeatureset{} \subseteq \symfeatureset{}$ and a bound $\symtargetsamplesize{} \in \mathbb{Z}_{\geq 0}$ and 
have to decide whether there is a sample with $t$-wise coverage and at most $\symtargetsamplesize{}$ configurations.
The problem can be cast as a type of \textsc{Set Cover} problem that has both its universe and its sets hidden behind a \symsat{} problem.
It is therefore unsurprising that \pisp{} is \NP-hard.
Furthermore, for a given formula $\symfeaturemodel{}$ with at least two variables in CNF,
deciding whether $\symtargetsamplesize{} = 0$ configurations suffice to achieve pairwise coverage of all interactions corresponds to the \textsc{Unsat} problem.
It is therefore also clear that \pisp{} is \coNP-hard, which means that the problem cannot be in \NP{} unless $\NP = \coNP$.
Furthermore, there cannot be any polynomial-time algorithm that approximates the number of configurations, unless $\P = \NP$.

To establish membership of \tisp{} in a complexity class,
we consider classes higher up in the \emph{polynomial-time hierarchy} \PH.
We establish the following result.
\begin{restatable}{theorem}{oraclelogqueries}
    \label{thm:sat-oracle-log-queries}
    Given a polynomial-time \symsat{} oracle $\symsatoracle{}$,
    for any constant $t$, \tisp{} with $|\symconcretefeatureset{}|$ concrete features 
    can be solved in polynomial time using $\mathcal{O}(\log |\symconcretefeatureset{}|)$ queries to $\symsatoracle{}$.
\end{restatable}
The full proof is based on first identifying the number of feasible interactions
using binary search and then requesting coverage of at least that many interactions by at most $\symtargetsamplesize{}$ configurations;
see \cref{sec:proofs-complexity}.
This establishes membership of \tisp{} in $\P^{\NP} = \Delta_2^p$ for any constant $t$.
Because $|\symconcretefeatureset{}|$ is bounded by the input size,
it also establishes membership in $\P^{\NP[\log]}$, the class of all 
problems that a polynomial-time Turing machine can solve with $\mathcal{O}(\log N)$ oracle queries for some \NP-complete problem, where $N$ is the input size.
By a result of Hemachandra~\cite{DBLP:journals/jcss/Hemachandra89}, this implies that it is possible to solve the problem using polynomially many \emph{non-adaptive} oracle queries, i.e., oracle queries that do not depend on the outcome of previous queries.
Another consequence of this is that the problem is unlikely to be $\Delta_2^p$-hard.

However, we show that even \pisp{} is hard for the \emph{Boolean hierarchy} \BH, the smallest superclass of \NP{} that is closed under complement, union and intersection.
\BH{} is known to be equal to $\QH = \bigcup_{k \in \mathbb{N}} \P^{\NP[k]}$~\cite{DBLP:journals/siamcomp/Wagner90}, the class of problems solvable by a deterministic polynomial-time Turing machine with any constant number of queries to an oracle for an \NP-complete problem.
This makes it unlikely that we can reduce the $\mathcal{O}(\log |\symconcretefeatureset{}|)$ queries in \cref{thm:sat-oracle-log-queries} to $\mathcal{O}(1)$:
if we were able to solve \pisp{} with $\mathcal{O}(1)$ queries, \QH{} and \BH{} would collapse to some finite level, 
which would in turn cause the collapse of \PH{} to its third level by a result of Kadin~\cite{DBLP:journals/siamcomp/Kadin88,DBLP:journals/siamcomp/Kadin91}.

\begin{restatable}{theorem}{bhhardness}
  \label{thm:bh-hardness}
  \pisp{} is \BH-hard.
\end{restatable}
The proof can be found in \cref{sec:proofs-complexity}.

\subsection{Avoiding Concrete Features.}
Our \BH-hardness reduction makes use of the fact that we are allowed to specify a set $\symconcretefeatureset{}$ of concrete features, 
restricting the features whose interactions need to be covered to that subset.
In this section, we show that several different variants of the problem, including one
where $\symconcretefeatureset{} = \symfeatureset{}$ is fixed, are equivalent under polynomial-time reductions.

\begin{definition}
  An interaction is a \emph{false interaction} if all its literals are negated variables, 
  a \emph{true interaction} if they are all non-negated, and a \emph{mixed interaction} otherwise.
\end{definition}

\begin{restatable}{theorem}{variantsreduction}
  \label{thm:variants-reduction}
  The following problems are polynomial-time equivalent for any $t \geq 2$:
  \begin{itemize}
    \item \tisp{},
    \item \tispac{}, which is \tisp{} with $\symconcretefeatureset{} = \symfeatureset{}$,
    \item \tispot{}, which is \tisp{} where only true interactions need coverage, and
    \item \tispacot{}, which is \tispac{} where only true interactions need coverage.
  \end{itemize}
\end{restatable}
The proof can be found in \cref{sec:proofs-complexity}.

\subsection{Hardness for Larger $t$.}
Finally, we also show that our hardness result extends from $t = 2$ to arbitrary constant $t \geq 2$;
again, due to space constraints, for the full proof, we refer to \cref{sec:proofs-complexity}.
\begin{restatable}{theorem}{hardnesslargert}
  For any $2 \leq t' \leq t$, we have $t'\textsc{-Isp-Ot} \leq_m^P \tispot$.
\end{restatable}

\section{Initial Heuristic.}
\label{sec:initial-heuristic}

We introduce a new initial solution heuristic that combines ideas from YASA~\cite{KTS+:VaMoS20} and the core-based approach of Yamada et al.~\cite{yamada2016greedy}.
YASA is an IPOG-based method that uses an incremental \symsat{} solver as a black box to maintain and incrementally construct multiple configurations simultaneously.
The intermediate partial configurations YASA maintains are always \emph{feasible}, i.e., can always be extended to a valid configuration.
The approach of Yamada et al.~\cite{yamada2016greedy} constructs one test at a time directly on the trail of an incremental CDCL solver using UP and clause learning.
While YASA offers flexibility by constructing multiple configurations simultaneously, it relies on frequent full \symsat{} calls to ensure feasibility.
In contrast, Yamada et al.\ exploit low-level solver access for efficiency, avoiding expensive satisfiability checks but limiting coverage flexibility.
Our heuristic combines both strengths: we maintain multiple partial configurations like YASA, but guide their construction using UP and clause learning in the style of Yamada et al., 
resulting in a scalable algorithm that produces high-quality test suites.

\subsection{Overview.}
At the core of our approach is a custom incremental \symsat{} solver designed to handle multiple partial configurations simultaneously.
Each partial configuration is represented as an independent \emph{trail} (a stack of literals).
Unlike standard \symsat{} solvers, which operate on a single trail, our solver maintains many, enabling clause sharing to reduce redundancy and memory overhead.
For a given trail~$\symtrail{}$, the solver provides two key operations:

\begin{description}
    \item[$\symtrail{}.\texttt{push\_and\_propagate}(\syminteraction{})$] Adds literals of interaction $\syminteraction{}$ (if not already present) and propagates them. 
      Returns \textbf{true} if successful, and \textbf{false} if a conflict occurs, using conflict resolution and backjumping to revert (at least) the push.
    \item[$\symtrail{}.\texttt{complete}()$] Attempts to extend $\symtrail{}$ to a full configuration of the feature model $\symfeaturemodel{}$.
      Returns \textbf{true} if no prior assignments needed to change to resolve conflicts, \textbf{false} otherwise.
\end{description}

Both operations automatically perform conflict resolution when necessary, which may involve learning clauses and non-local updates to the trail, including the revision of earlier assignments.
The returned trail is always conflict-free.
Note that operations similar to these could also be performed by using existing incremental SAT solvers with a sufficiently rich interface, such as CaDiCaL~\cite{BiereFallerFazekasFleuryFroleyks-CAV24}, by tracking a set of assumptions for each partial configuration.
\texttt{push\_and\_propagate} can almost be simulated by \texttt{assume} and \texttt{propagate} in CaDiCaL (except for the conflict resolution),
and \texttt{complete} is similar to \texttt{solve} with assumptions, but does not produce a complete solution that potentially ignores some assumptions in case the assumptions are unsatisfiable.
Furthermore, on one hand, very frequently switching between completely different sets of assumptions would likely be too inefficient.
On the other hand, using one instance of such a SAT solver per partial configuration wastes memory and hampers sharing of learnt clauses.

Our heuristic iteratively constructs a set~$\symworkingsample{}$ of trails, i.e., partial configurations, that together aim to cover the set~$\symfeasibleinteractions{}$ of interactions that have not yet been proved infeasible.
In each iteration, it selects up to~$\symtargetqueuesize{}$ uncovered interactions and attempts to assign each to an existing trail in~$\symworkingsample{}$ using \cref{alg:extend-configurations}.

\begin{algorithm}
\caption{\texttt{ExtendConfs}$(\symworkingsample{}, \syminteraction{}) \rightarrow$ bool}\label{alg:extend-configurations}
\begin{algorithmic}[1]
\FOR{$\symtrail{} \in \text{sorted}_\syminteraction{}(\symworkingsample{})$}
    \IF{$\symtrail{}.\texttt{push\_and\_propagate}(\syminteraction{})$}
        \RETURN \textbf{true}
    \ENDIF
\ENDFOR
\RETURN \textbf{false}
\end{algorithmic}
\end{algorithm}

If no trail in~$\symworkingsample{}$ can accommodate~$\syminteraction{}$, a new trail is created and initialized with~$\syminteraction{}$.
If this also fails, the interaction is marked as infeasible and excluded from further consideration.
Once no uncovered interactions remain, the algorithm attempts to \emph{complete} all trails into full configurations using \cref{alg:complete-configurations}.

\begin{algorithm}
\caption{\texttt{Complete}$(\symworkingsample{}) \rightarrow$ bool}\label{alg:complete-configurations}
\begin{algorithmic}[1]
\FOR{$\symtrail{} \in \symworkingsample{}$}
    \IF{ $\neg \symtrail{}.\texttt{complete}()$}
        \RETURN \textbf{false}
    \ENDIF
\ENDFOR
\RETURN \textbf{true}
\end{algorithmic}
\end{algorithm}

If all trails complete successfully,~$\symworkingsample{}$ forms a valid solution. 
Otherwise, the algorithm continues iterating, revisiting interactions uncovered by the completion operation.
The complete high-level process is summarized in \cref{alg:initial-solution-heuristic}.

\begin{algorithm}
\caption{Initial solution heuristic (simplified)}
\label{alg:initial-solution-heuristic}
\begin{algorithmic}[1]
        \renewcommand{\COMMENT}[1]{\hfill\textit{// #1}}
\STATE $\symworkingsample{} \gets \emptyset$, $\symfeasibleinteractions{} \gets$ all interactions
\WHILE{true}
    \STATE $\symqueue{} \gets$ up to $\symtargetqueuesize{}$ uncovered interactions from $\symfeasibleinteractions{}$
    \IF{$\symqueue{} = \emptyset \wedge \texttt{Complete}(\symworkingsample{})$}
        \RETURN $\symworkingsample{}$
    \ENDIF
    \FOR{$\syminteraction{} \in sorted_\symworkingsample{}(\symqueue)$}
        \IF{not \texttt{ExtendConfs}$(\symworkingsample{}, \syminteraction{})$}
            \STATE Create empty trail $\symtrail{}$
            \IF{not $\symtrail{}.\texttt{push\_and\_propagate}(\syminteraction{})$}
                \STATE $\symfeasibleinteractions{}\gets \symfeasibleinteractions{}\setminus\{\syminteraction{}\}$ \COMMENT{Mark $\syminteraction{}$ as infeasible}
            \ELSE
                \STATE $\symworkingsample{} \gets \symworkingsample{} \cup \{\symtrail{}\}$
            \ENDIF
        \ENDIF
    \ENDFOR
\ENDWHILE
\end{algorithmic}
\end{algorithm}

\subsection{Implementation Details.}
Our implementation includes many low-level optimizations and other details for high performance; see \cref{sec:details-heuristic} and the accompanying source code for details.
The set \( \symworkingsample{} \) in \cref{alg:extend-configurations} is sorted to favor configurations with high overlap with \( \syminteraction{} \).
The set \( \symqueue{} \) in \cref{alg:initial-solution-heuristic} is implemented as priority queue sorted by compatibility with \( \symworkingsample{} \),
prioritizing interactions with few available candidate configurations; when $\symworkingsample{}$ changes, \( \symqueue{} \) is updated.
The target size $\symtargetqueuesize{}$ increases exponentially up to some bound.
\( \symqueue{} \) is populated with a random sample of uncovered interactions to leverage implicit coverage;
maintaining priorities over the full interaction set would require too much memory for large instances.
\( \symqueue{} \) thus serves as a dynamically maintained shortlist of promising candidates,
which is essential to avoid tracking coverage of each interaction and enumerating uncovered interactions too often,
both of which are infeasible for huge instances.

As the \emph{initial phase} of our algorithm, we apply the above heuristic multiple times, interleaved with the lower bound heuristic, to potentially yield improved solutions.
In the first application, we initialize $\symworkingsample{}$ using a greedy configuration that favors negated features, 
and initialize \( \symqueue{} \) with heuristically selected positive literal pairs.
In subsequent applications, we start with an empty $\symworkingsample{}$ and initialize \( \symqueue{} \) based on our lower bound heuristic; see~\cref{sec:initial-lower-bounds}.
Feasible interactions and learned clauses are cached to accelerate subsequent applications.

\subsection{Preprocessing.}
\label{sec:preprocessing}
Preprocessing \symsat{} formulas has become essential, as the raw CNF representations automatically generated in many real-world applications are often far from optimal and can typically be significantly reduced.
In \symsat{} preprocessing, a given formula $\symfeaturemodel{}$ is transformed into a new formula $\symfeaturemodel{}'$ that is equisatisfiable but usually more compact and easier to solve.
In this section, we briefly discuss how we preprocess a formula \symfeaturemodel{} such that we can map samples of the resulting $\symfeaturemodel{}'$ back to \symfeaturemodel{}, preserving their size and $t$-wise coverage.
Our implementation only preprocesses once, before the first feasible sample is computed, and all algorithms are then applied to the preprocessed model.

Various preprocessing techniques are used in practice in \symsat{} preprocessing;
however, not all are suitable for our problem, as some techniques preserve only equisatisfiability.
Consider, for instance, \emph{pure literal elimination}, which eliminates a variable that only occurs positively (or only negatively) by setting it to the corresponding value, also eliminating all clauses it appears in.
This is a safe operation on non-concrete features, but clearly must not be performed on concrete features.
Thus, our problem requires preprocessing that preserves logical equivalence with respect to the set of feasible concrete interactions.
Incremental \symsat{} solvers such as CaDiCaL~\cite{BiereFallerFazekasFleuryFroleyks-CAV24} support \emph{freezing} variables to protect them from logical changes during preprocessing, making their use safe in our setting if all concrete features are frozen.
We could thus use their preprocessing pipeline for our problem.
However, that also prohibits some preprocessing operations on concrete features that would be safe.
For instance, interactions involving equivalent features are themselves equivalent.
Consequently, we implemented a custom preprocessing pipeline, primarily following established techniques~\cite{biere2021preprocessing} but adapting them to our specific requirements.

We employ \emph{failed and equivalent literal detection}.
This can eliminate many implied interactions from $\symfeasibleinteractions{}$, as the interactions of equivalent literals are themselves equivalent.
This substitution is safe in our setting as long as at least $t$ concrete features are preserved; otherwise, $\symfeasibleinteractions{}$ would become empty.
We also employ \emph{bounded variable elimination} (BVE).
For \symsat{} solvers, BVE~\cite{een2005effective} is among the most effective preprocessing techniques.
However, it eliminates variables and only preserves satisfiability; in our case, it can only be safely applied to non-concrete features, whose interactions need not be covered.
Finally, we also employ \emph{clause vivification}~\cite{LI2020103197} and removal of subsumed clauses.
These standard techniques preserve logical equivalence and are thus safe to use in our context.

The implication graph, used for detecting equivalent literals, can also be exploited to identify implied interaction coverage, thereby further reducing $\symfeasibleinteractions{}$.
For example, from the implications $\ell_1 \rightarrow \ell_2$ and $\ell_2 \rightarrow \ell_3$, we can infer that covering $(\ell_1, \ell_2)$ implies coverage of both $(\ell_2, \ell_3)$ and $(\ell_1, \ell_3)$.
Thus, it suffices to retain only $(\ell_1, \ell_2)$ in $\symfeasibleinteractions{}$.
More generally, if $\UP(\{\ell_1,\ell_2\})$ contains some other interaction $\syminteraction{}$, $\syminteraction{}$ can also be removed from $\symfeasibleinteractions{}$.
We refer to the removal of such interactions as \emph{universe reduction}.
Although it does not alter the formula, it can significantly decrease the number of interactions that must be explicitly considered.

\section{Lower Bounds.}
\label{sec:lower-bounds}
After identifying the set of feasible interactions $\symfeasibleinteractions{}$ in the first application of our initial heuristic,
we introduce binary clauses as needed to ensure that UP detects all infeasible interactions, 
i.e., that if $\{\ell_1, \ell_2\}$ is an infeasible interaction, $\overline{\ell_2} \in \UP(\ell_1)$ and $\overline{\ell_1} \in \UP(\ell_2)$.
We then compute initial lower bounds on the number of required configurations as follows.
We consider feasible interactions to be vertices of a graph $G$, in which two interactions $\syminteraction{}, \syminteraction{}'$ are connected by an edge if there is no valid configuration $\symconfiguration{} \supseteq \syminteraction{} \cup \syminteraction{}'$.
Note that any clique $\symmaxexclusiveset{}$ of $G$ induces a lower bound, 
i.e., $|\symmaxexclusiveset{}| \leq |\symsample{}|$ for any sample $\symsample{}$ with pairwise coverage,
because each configuration $\symconfiguration{} \in \symsample{}$ can cover at most one interaction $\syminteraction{} \in \symmaxexclusiveset{}$.

Because $G$ is often very large and difficult to compute, we rely on the subgraph $G_2$ of $G$ which contains an edge iff $\UP(\syminteraction{} \cup \syminteraction{}') = \bot$.
By construction, most relevant graph operations on $G_2$ can be done using UP.
In particular, we never have to compute the edge set of this graph explicitly.
Because $G_2$ is a subgraph of $G$, its cliques also induce lower bounds.

\subsection{Initial Lower Bounds.}
\label{sec:initial-lower-bounds}
After the first application of the solution heuristic outlined in \cref{sec:initial-heuristic}, and after each application thereafter, we compute a clique on $G_2$ as follows.
During the initial heuristic, we track which interactions are the first to be inserted into each partial configuration in $\symworkingsample{}$;
such interactions are \emph{spawners}.
To compute an initial clique, we use a naive clique algorithm to find a maximal clique on $G_2[P]$,
where $P$ is either the set of all spawners up to this point or just the spawners during the last application of our initial heuristic.
In each step, a vertex is selected uniformly at random among all remaining candidates and added to the clique.
We repeat the process several times for each $P$;
the best clique found is kept, potentially updating the lower bound, and used as initial $\symqueue{}$ in the next application of the initial heuristic.

\subsection{Cut \& Price Bounds.}
After the initial phase, we apply a linear programming-based algorithm, which combines cut \& price and rounding techniques, to find cliques on $G_2 = (\symfeasibleinteractions{}, E_2)$ as lower bounds.
Here, we first give a high-level description of our approach before describing its components in more detail.

We use an 
integer programming (IP) formulation of the clique problem on $G_2 = (\symfeasibleinteractions{}, E_2)$,
with a variable $x_\syminteraction{} \in \{0,1\}$ for each interaction $\syminteraction{} \in \symfeasibleinteractions{}$ and a constraint for each independent set $\symindset{}$ of $G_2$ enforcing that at most one $\syminteraction{} \in \symindset{}$ is selected.
As working with the full set of variables would be practically infeasible due to the size of $\symfeasibleinteractions{}$, 
we work with a dynamic subset $\symfeasibleinteractions{}' \subseteq \symfeasibleinteractions{}$ of interactions.
Similarly, the set of constraints we use is induced by a dynamic subset $\symindsetfamily{}'$ of the independent sets $\symindsetfamily{}$ of $G_2$.

We then repeatedly solve the linear relaxation of our IP, using a greedy rounding scheme to obtain new, potentially better cliques from the relaxed solution.
We then strengthen or add constraints to cut off the current relaxed solution, or use pricing to introduce new interactions to $\symfeasibleinteractions{}'$, before solving the next relaxation.

\subsubsection{Greedy Rounding.}
After finding the optimal solution $x^*$ of the current relaxation, we round as follows.
Starting with an empty clique $\symmaxexclusiveset{}$, we iterate through all $\syminteraction{} \in \symfeasibleinteractions{}'$ with $x_\syminteraction{}^* > 0$ in order of non-increasing value $x_\syminteraction{}^*$ in the relaxation; 
we add $\syminteraction{}$ to $\symmaxexclusiveset{}$ if $\syminteraction{}$ is adjacent to all previously added $J \in \symmaxexclusiveset{}$.
Finally, we make the resulting clique maximal by adding interactions adjacent to all $J \in \symmaxexclusiveset{}$ from $\symfeasibleinteractions{}'$ chosen uniformly at random.
If that results in a better clique $\symmaxexclusiveset{}$, we record that clique.

\subsubsection{Constraints.}
Note that each (partial) configuration $\sympartialconfig{}$ with $\UP(\sympartialconfig{}) \neq \bot$ induces an independent set $\symindset{}(\sympartialconfig{}) = \{\syminteraction{} \in \symfeasibleinteractions{} \mid \syminteraction{} \subseteq \sympartialconfig{}\}$ of $G_2$.
We can thus turn (partial or complete) configurations 
into constraints and use trails to represent constraints internally.
We initialize $\symindsetfamily{}'$ with the best initial sample $\symsample{}$, ensuring that our LP relaxation solution has value at most $|\symsample{}|$.

\subsubsection{Variables and Pricing.}
We initialize $\symfeasibleinteractions{}'$ to contain the best clique on $G_2$ found in the initial phase,
and either all spawners or the spawners from the best run of the initial heuristic, depending on the sizes of those sets.
We mainly use the linear relaxation, where we have $x_\syminteraction{} \geq 0$ instead of $x_\syminteraction{} \in \{0,1\}$.
This relaxation has the following dual.
\begin{align*}
  \min{} \sum\limits_{\symindset{} \in \symindsetfamily{}} z_\symindset{}&\text{ s.t.}\\
    \forall \syminteraction{} \in \symfeasibleinteractions{}: \sum\limits_{\symindset{} \in \symindsetfamily{}, \syminteraction{} \in \symindset{}} z_\symindset{} &\geq 1,\\
    z_\symindset{} &\geq 0.
\end{align*}
Let $o^*$ be the objective value of the current relaxation.
For any given subset $\symfeasibleinteractions{}'$, $\lfloor o^* \rfloor$ is an upper bound on the size of any clique of $G_2[\symfeasibleinteractions{}']$.
We compute this bound, including some buffer for numerical errors before rounding down;
we detect that the currently found clique $\symmaxexclusiveset{}$ is optimal for $\symfeasibleinteractions{}'$ if $|\symmaxexclusiveset{}| = \lfloor o^* \rfloor$.

In that case, we continue by \emph{pricing}, i.e., searching for new interactions to add to $\symfeasibleinteractions{}'$ that have the potential to improve $x^*$, given the current $\symindsetfamily{}'$.
Before we describe the pricing in detail, we make the following crucial remark.
We internally represent the independent set $\symindset{}$ underlying some constraint using a trail, which represents a set of non-conflicting literals $\sympartialconfig{}$ closed under UP.
In our internal representation, we therefore already include some interactions $\syminteraction{} \notin \symfeasibleinteractions{}'$ in our constraints:
conceptually, $\syminteraction{}$ is included in $\symindset{}$ iff $\syminteraction{} \subseteq \sympartialconfig{}$, even if $\syminteraction{} \notin \symfeasibleinteractions{}'$.

Together with the primal solution $x^*$, we also obtain a solution $z^*$ of the dual; this solution assigns a weight $z_\symindset{} \geq 0$ to each independent set $\symindset{} \in \symindsetfamily{}'$.
This solution has $\sum_{\symindset{} \in \symindsetfamily{}'}z_\symindset{} = o^*$ due to strong duality and is feasible for the current dual, which contains a constraint for each $\syminteraction{} \in \symfeasibleinteractions{}'$.
Consider the dual problem that we obtain by extending the subset $\symfeasibleinteractions{}'$ to the full set of feasible interactions $\symfeasibleinteractions{}$.
If $z^*$ is feasible for this extended dual problem, i.e., if there is no $\syminteraction{} \in \symfeasibleinteractions{} \setminus \symfeasibleinteractions{}'$ with $\sum_{\symindset{} \ni \syminteraction{}} z_\symindset{} < 1$,
then, by strong duality, $\lfloor o^* \rfloor$ is an upper bound on the clique size of the entire graph $G_2$.

To find interactions that may potentially lead to better cliques, it thus suffices to find interactions $\syminteraction{}$ with $\sum_{\symindset{} \ni \syminteraction{}} z_\symindset{} < 1$.
Instead of iterating through the full set of interactions, we begin by pricing on smaller sets of interactions.
We begin by trying $P_a$, the set of all \emph{spawners} encountered during our initial heuristic.
If that yields no violated dual constraint, we next consider the set of all interactions that were taken from the priority queue $\symqueue{}$ and explicitly pushed to one of the trails during the last iteration of our initial heuristic.
If that still finds no violated dual constraint, we instead price a random sample of $\symfeasibleinteractions{}$.
Only if all previous steps fail, we fully enumerate and price $\symfeasibleinteractions{}$.
In any case, we apply a limit on the number of interactions we introduce at once.

\subsubsection{Cutting Planes.}
If we have not established optimality on $\symfeasibleinteractions{}'$ 
after attempting to obtain a better clique by rounding, 
we need a way to make progress 
beyond the current relaxed solution $x^*$.
A primary way 
is to \emph{tighten} the relaxation by adding new constraints violated by $x^*$ 
or by strengthening existing ones.

To this end, we first scan $x^*$ for \emph{violated non-edges}, i.e., pairs of interactions $\syminteraction{}, \syminteraction{}'$ with $\UP(\syminteraction{} \cup \syminteraction{}') \neq \bot$ and $x_\syminteraction{}^* + x_{\syminteraction{}'}^* > 1$.
During each round of tightening, we generate a list of all violated non-edges.
We forbid them by strengthening constraints or adding new ones.
To strengthen an existing constraint, we scan through incomplete configurations in our constraints and check,
for each configuration in which none of the involved literals are false, whether we can push both interactions to the corresponding trail;
otherwise, we have to create new trails and corresponding constraints.
We always prefer strengthening existing constraints over introducing new constraints.
We take care to treat each violated non-edge at most once;
furthermore, to reduce the number of constraints generated, 
we generate at most one new constraint for each interaction $\syminteraction{}$ involved in violated non-edges in the current round of tightening.

If the relaxed solution $x^*$ has no violated non-edges, we have several options to continue.
Firstly, we can run pricing; while this does not tighten the relaxation, it may still take us away from the current $x^*$ to one that has violated non-edges.
It can also help avoiding focusing too much on some subset $\symfeasibleinteractions{}'$, which may not contain the best clique after all;
we thus run pricing after every \num{40} iterations that did not encounter violated non-edges.

If we do not opt for pricing, we first use the following greedy strengthening approach.
For each current constraint $\symindset{}$, we iterate through all interactions $\syminteraction{} \in \symfeasibleinteractions{}'$ with non-zero $x_\syminteraction{}^*$ in order of non-increasing $x_\syminteraction{}^*$, 
summing up the weights $x_\syminteraction{}^*$ of all $\syminteraction{}$ that do not contradict $\symindset{}$ without actually changing $\symindset{}$.
If the resulting value indicates that $\symindset{}$ could become a violated constraint, we attempt to actually expand $\symindset{}$.
This is done by again iterating over $\syminteraction{}$ in the same order, this time greedily pushing each interaction $\syminteraction{}$ into $\symindset{}$ unless that leads to a conflict;
we do not undo these changes, even if we do not end up with a violated constraint.
If at least one violated constraint resulted from the strengthening, we continue by solving the new relaxation.

Otherwise, we attempt to generate a new violated constraint using a similar greedy approach.
Again considering the list of interactions $\syminteraction{}$ with non-zero $x_\syminteraction{}^*$ in non-increasing order, from each index in that list we start to construct a potential new independent set as follows.
Beginning with the starting index, we push interactions to an initially empty propagator, unless that causes a conflict; on reaching the end of the list, we resume at the start.
We record all independent sets that would result in violated constraints, together with their total right-hand side value;
we add up to \num{10} most strongly violated constraints to $\symindsetfamily{}'$.
If we generated at least one violated constraint, we continue by solving the new relaxation.

As a final heuristic attempt at finding violated constraints, we re-run the initial heuristic, starting with the best clique found so far.
We then check to see if any of the resulting configurations can be used as violated constraint.
If this fails as well, we resort to pricing, including a full pricing pass through $\symfeasibleinteractions{}$ if necessary; 
if pricing also does not yield new interactions, our algorithm gets stuck on the current relaxation $x^*$, and we abort the search.

\section{Main Algorithm.}
\label{sec:sammy}
Our algorithm, Sammy, follows up on the initial heuristic with a parallel Large Neighborhood Search (LNS) heuristic that builds on the SampLNS algorithm~\cite{10.1145/3712193}.
As outlined in \cref{alg:sammy}, we begin by preprocessing the input formula~$\symfeaturemodel{}$ using the techniques described in \cref{sec:preprocessing}, yielding a simplified formula~$\symfeaturemodel{}'$.
We then compute an initial heuristic solution~$\symsample{}$ and the corresponding set of covered interactions~$\symfeasibleinteractions{}$, as described in \cref{sec:initial-heuristic}.

\begin{algorithm}
  \caption{Sammy$(\symfeaturemodel{}, \symnumthreads{})$}
\label{alg:sammy}
\begin{algorithmic}[1]
    \renewcommand{\COMMENT}[1]{\hfill\textit{// #1}}
    \STATE $\symfeaturemodel{}' \gets \texttt{Preprocess}(\symfeaturemodel{})$ \COMMENT{\cref{sec:preprocessing}}
    \STATE $\symsample{}, \symfeasibleinteractions{} \gets$ initial sample and interactions (\cref{sec:initial-heuristic})
    \STATE $\symmaxexclusiveset{}, \symsample{} \gets$ repeated initial LB/sample (\cref{sec:initial-lower-bounds})
    \STATE $\symdestroyparameters{} \gets$ initial destroy parameters
    \STATE \textit{// Shared: } $\symsample{}, \symdestroyparameters{}, \symmaxexclusiveset{}, \symfeaturemodel{}'$
    \STATE \textbf{spawn} lower-bound worker updating $\symmaxexclusiveset{}$ (\cref{sec:lower-bounds})
    \IF{$|\symfeasibleinteractions{}|\leq \num{100000}$}
    \STATE \textbf{spawn} full solution worker (\cref{sec:full-prob-solver})
    \ENDIF
    \WHILE{$|\symmaxexclusiveset{}| < |\symsample{}|$}
        \STATE Initialize channel \texttt{result}
        \FOR{$i = 1$ to $\symnumthreads{}$}
            \STATE \textit{// Speculative Parallelism}
            \STATE \textbf{spawn} \texttt{DestroyAndRepair}$(\symsample{}, \symfeaturemodel{}', \symdestroyparameters{}, \texttt{result})$
        \ENDFOR
        \STATE Wait for first $\widehat{\symsample{}}$ from \texttt{result}
        \STATE Signal all threads from line~12 to terminate
        \STATE $\symsample{} \gets \widehat{\symsample{}}$
    \ENDWHILE
    \STATE Postprocess $\symsample{}$ to map back to original features
    \RETURN $\symsample{}, |\symmaxexclusiveset{}|$ \COMMENT{Solution and lower bound}
\end{algorithmic}
\end{algorithm}

\begin{algorithm}
\caption{\texttt{DestroyAndRepair}$(\symsample{}, \symfeaturemodel{}', \symdestroyparameters{}, \texttt{result})$}
\label{alg:destroy-and-repair}
\begin{algorithmic}[1]
    \renewcommand{\COMMENT}[1]{\hfill\textit{// #1}}
    \STATE \textit{// Replace a $\symsample{}_d \subseteq \symsample{}$ with $\symsample{}_r $ such that $|\symsample{}_r | < |\symsample{}_d |$}
    \WHILE{not terminated}
        \STATE $\symsample{}' \gets \texttt{Destroy}(\symsample{}, \symdestroyparameters{})$ \COMMENT{$\symsample{}' \subsetneq \symsample{}$}
        \STATE $\symsample{}_d \gets \symsample{} \setminus \symsample{}'$ \COMMENT{Destroyed configurations}
        \STATE $\symfeasibleinteractions{}' \gets \symfeasibleinteractions{} \setminus \symfeasibleinteractions{}(\symsample{}')$ \COMMENT{Missing interactions}
        \STATE $\symmaxexclusiveset{}' \gets \texttt{MutExclSet}(\symfeasibleinteractions{}', \symsample{} \setminus \symsample{}', \symfeaturemodel{}')$ \COMMENT{\cref{sec:lower-bounds}}
        \IF{$|\symmaxexclusiveset{}'| = |\symsample{}'|$}
            \STATE \textbf{continue} \COMMENT{Already optimal}
        \ENDIF
        \STATE Select random repair parameters $\symrepairparameters{}$
        \STATE $\symsample{}_r \gets \texttt{Repair}(\symfeaturemodel{}', \symfeasibleinteractions{}', \symmaxexclusiveset{}', |\symsample{}_d|-1, \symrepairparameters{})$
        \STATE Update $\symdestroyparameters{}$ based on performance
        \IF{$|\symsample{}_r | < |\symsample{}_d |$}
            \STATE \texttt{result} $\gets \symsample{}_r \cup \symsample{}'$
            \STATE \textbf{break}
        \ENDIF
    \ENDWHILE
\end{algorithmic}
\end{algorithm}

We spawn two background threads: one continuously maximizes and updates the best mutually exclusive set $\symmaxexclusiveset{}$ (see \cref{sec:lower-bounds});
the other attempts to compute a complete solution without using LNS, which typically succeeds only on small instances, and, thus, is only used if $|\symfeasibleinteractions{}|\leq \num{100000}$.
The main loop runs while the current lower bound~$|\symmaxexclusiveset{}|$ is smaller than the size of the sample~$|\symsample{}|$.

In each iteration, we initialize a result channel and spawn $\symnumthreads{}$ threads, each executing the \texttt{DestroyAndRepair} procedure (\cref{alg:destroy-and-repair}).
Each thread attempts to remove a subset of configurations from~$\symsample{}$ via the \texttt{Destroy} function.
The size of the destroyed set~$\symsample{}_d $ is governed by the destroy parameters~$\symdestroyparameters{}$, which are updated based on the performance of the repair step, similar to SampLNS~\cite{10.1145/3712193};
for more details, see \cref{sec:a-lns}.
The interactions that are no longer covered are identified as~$\symfeasibleinteractions{}'$, and a new exclusive interaction set~$\symmaxexclusiveset{}'$ is computed for~$\symfeasibleinteractions{}'$, providing a lower bound and symmetry breaker for the subproblem.
This follows the same methodology described in \cref{sec:lower-bounds}, but restricted to the reduced interaction set~$\symfeasibleinteractions{}'$.
If~$|\symmaxexclusiveset{}'| = |\symsample{}'|$, then the subproblem is already optimally solved and the thread skips the repair.

Otherwise, the thread randomly selects repair parameters~$\symrepairparameters{}$ and invokes the \texttt{Repair} procedure, trying to find a new sample~$\symsample{}_r $ that covers all interactions in~$\symfeasibleinteractions{}'$ such that $|\symsample{}_r | \leq |\symsample{}_d |-1$.
These randomized parameters may trigger highly diverse repair strategies, as detailed in \cref{sec:repair-strategies}.
If the resulting set~$\symsample{}_r $ is smaller than~$\symsample{}_d $, we communicate the improved sample~$\widehat{\symsample{}} = \symsample{}_r \cup \symsample{}'$.

Upon receiving the first improved sample~$\widehat{\symsample{}}$, all remaining threads are interrupted and $\symsample{}$ is updated.
This process repeats until~$|\symsample{}| = |\symmaxexclusiveset{}|$, in which case the sample is provably optimal.
Alternatively, the algorithm may terminate due to a time limit or because the full solution worker finds an optimal sample (not shown in the pseudocode for clarity).
Finally, the sample~$\symsample{}$ is postprocessed to map it back to the original feature space, and the sample and lower bound are returned.

\subsection{Repair Strategies.}
\label{sec:repair-strategies}
\label{sec:exact-algorithm}
Here, we discuss our different repair strategies.
In any case, we obtain as input a formula $\symfeaturemodel{}'$, a set $\symfeasibleinteractions{}'$ of valid interactions, 
a mutually exclusive set $\symmaxexclusiveset{}' \subseteq \symfeasibleinteractions{}'$ and a sample $\symsample{}_d $ covering $\symfeasibleinteractions{}$',
and our goal is to find such a sample of size at most $\symtargetrepairsize := |\symsample{}_d | - 1$, or to prove that no such sample exists.

\subsubsection{Core Model.}
At the core of the algorithm is a \symsat{} model of \pisp.
Similar to existing approaches like~\cite{10.1145/3712193}, which used a CP-SAT model,
we model the existence of a sample covering $\symfeasibleinteractions{}'$ with at most $\symtargetrepairsize{}$ configurations as a \symsat{} formula.
For each $1 \leq i \leq \symtargetrepairsize{}$, we have a variable $x_j^i$ for each feature $x_j \in \symfeatureset{}$;
each clause $\gamma$ in $\symfeaturemodel{}$ results in $\symtargetrepairsize{}$ clauses $\gamma^i$ in our model by replacing all $x_j$ in $\gamma$ by $x_j^i$.
Essentially, our \symsat{} formula contains $\symtargetrepairsize{}$ independent copies of $\symfeaturemodel{}$,
allowing us to encode $\symtargetrepairsize{}$ independent configurations.

Furthermore, for each $1 \leq i \leq \symtargetrepairsize{}$ and each interaction $\syminteraction{} \in \symfeasibleinteractions{}'$,
we have a variable $y_\syminteraction{}^i$ indicating whether $\syminteraction{}$ is covered by the $i$th copy of $\symfeaturemodel{}$.
For $\syminteraction{} = \{\ell_1, \ell_2\}$, we add the clauses $y_\syminteraction{}^i \vee \overline{\ell_1^i} \vee \overline{\ell_2^i}$, $\overline{y_\syminteraction{}^i} \vee \ell_1^i$ and $\overline{y_\syminteraction{}^i} \vee \ell_2^i$,
where $\ell_h^i$ is the literal obtained by replacing any $x_j$ by $x_j^i$ in the literal $\ell_h$.
We enforce that each $\syminteraction{} \in \symfeasibleinteractions{}'$ is covered by introducing the clause $\bigvee_{i=1}^\symtargetrepairsize{} y_\syminteraction{}^i$.

Finally, for each literal $\ell \in \bigcup_{\syminteraction{} \in \symfeasibleinteractions{}'} \syminteraction{}$, 
we know that at least one configuration must set $\ell$ to true;
we introduce clauses $\bigvee_{i=1}^\symtargetrepairsize{} \ell^i$ ensuring that each literal occurring in $\symfeasibleinteractions{}'$ is true in at least one configuration.

We use the mutually exclusive set $\symmaxexclusiveset{}' = \{U_1, \ldots, U_m\}$ to break symmetries.
Because each $U_j$ must be in a different configuration, we can fix $U_i$ to be covered by the $i$th copy of $\symfeaturemodel{}$ for all $1 \leq i \leq m$.
To do this, we compute $\UP(U_j)$, removing variables and satisfied clauses and shortening remaining clauses as appropriate.
If all but one literal $\ell \in \syminteraction{}$ are fixed to true in the $i$th copy, we also replace $y_\syminteraction{}^i$ by $\ell$.

\subsubsection{Incrementality.}
The size of our \symsat{} formula is typically dominated by the large number of interactions $|\symfeasibleinteractions{}'|$.
Because many interactions are implicitly covered as a byproduct of covering other interactions,
we consider the following four strategies for handling $\symfeasibleinteractions{}'$.
In addition to starting with the full set $\symfeasibleinteractions{}'$ and solving the formula once \emph{non-incrementally},
we have three incremental strategies called \emph{simple incremental}, \emph{greedy incremental} and \emph{alternating LB-UB}.
Each incremental strategy maintains a growing subset $\symfeasibleinteractions{}'' \subseteq \symfeasibleinteractions{}'$ that is part of its \symsat{} formulation;
the difference in the strategies lies in how that set is grown and when and how the \symsat{} solver is called.
Note that adding interactions to $\symfeasibleinteractions{}''$ can be implemented by adding variables and clauses to the \symsat{} formulation,
making use of modern incremental \symsat{} solvers to reduce the runtime of repeated \symsat{} calls.

\paragraph{Simple Incremental Strategy.}
Starting with $\symfeasibleinteractions{}'' = \symmaxexclusiveset{}'$, in each iteration, we first solve the \symsat{} model with the current $\symfeasibleinteractions{}''$.
If all interactions from $\symfeasibleinteractions{}'$ are covered or the formula is \symunsat{}, we are done; otherwise, we compute the uncovered interactions $\mathcal{J}$.
If $\mathcal{J}$ is large compared to $\symfeasibleinteractions{}''$, we add a random subset of $\mathcal{J}$ to $\symfeasibleinteractions{}''$; otherwise, we add all of $\mathcal{J}$.
At least \qty{2.5}{\percent} of $\symfeasibleinteractions{}'$ is added in each iteration; random covered interactions are used if needed.
If $|\symfeasibleinteractions{}''|$ grows to beyond $0.33|\symfeasibleinteractions{}'|$, we instead extend $\symfeasibleinteractions{}''$ to $\symfeasibleinteractions{}'$.

\paragraph{Greedy Incremental Strategy.}
We can, of course, reuse our initial heuristic (\cref{alg:initial-solution-heuristic}) as a repair strategy.
However, a variant in which the partial configurations only grow can also be used to select $\symfeasibleinteractions{}'' \subseteq \symfeasibleinteractions{}'$ by collecting all explicitly covered interactions, 
i.e., those taken from $\symqueue{}$, in $\symfeasibleinteractions{}''$.
If the heuristic succeeds in covering all interactions using at most~$\symtargetrepairsize{}$ configurations, we can return the solution directly without a \symsat{} call.
Otherwise, as soon as \cref{alg:initial-solution-heuristic} would create the $(\symtargetrepairsize{}{+}1)$-st trail, we solve the \symsat{} formula to verify whether $\symfeasibleinteractions{}''$ truly cannot be covered using only~$\symtargetrepairsize{}$ configurations,
or whether the existing trails could be modified accordingly.

If the formula is unsatisfiable, we obtain a proof that~$\symfeasibleinteractions{}'$ cannot be covered with~$\symtargetrepairsize{}$ configurations and return.
If a satisfying assignment is found, we reassign the interactions in~$\symfeasibleinteractions{}''$ to trails accordingly and continue with our heuristic, 
resetting literals not implied by $\symfeasibleinteractions{}''$ in the process.
If the returned assignment covers all of~$\symfeasibleinteractions{}'$, we can immediately return it as a valid solution.

To improve convergence, before the \symsat{} call, we also add all uncovered interactions that fit into at most a single trail to~$\symfeasibleinteractions{}''$.
If this still does not increase the size of~$\symfeasibleinteractions{}''$ sufficiently compared to the previous \symsat{} call, we include additional interactions, prioritizing those with fewer candidate configurations.
As the incremental approach does not allow removing interactions from~$\symfeasibleinteractions{}''$, we treat all these added interactions as explicitly covered afterwards.

Since each interaction in~$\symmaxexclusiveset{}'$ requires its own trail, we directly initialize the trails with these interactions.
Moreover, the interactions in~$\symfeasibleinteractions{}''$, especially those that triggered the creation of new trails, 
also serve as candidates for constructing larger mutually exclusive sets $\symmaxexclusiveset{}'$, a property exploited in the next strategy.

\paragraph{Alternating LB-UB Strategy.}
This strategy behaves like \emph{greedy incremental}, but attempts to improve the lower bound $|\symmaxexclusiveset{}'|$ before each of our \symsat{} calls to improve the symmetry breaking,
by running an incremental variant of our cut, price \& round scheme to find such sets for a few iterations, focusing on the interactions in $\symfeasibleinteractions{}''$.
If this succeeds, we may have to recreate our \symsat{} model instead of simply adding to it in order to utilize the larger symmetry breaker. 
If necessary, we can also include interactions from $\symfeasibleinteractions{}' \setminus \symfeasibleinteractions{}''$ in this search, extending $\symfeasibleinteractions{}''$  
if we find a larger mutually exclusive set that includes some interaction not yet in $\symfeasibleinteractions{}''$.

\subsection{Full Problem Solver.}
\label{sec:full-prob-solver}
We also apply a variant of our alternating LB-UB strategy, which supports an incrementally growing number of configurations
instead of a fixed upper bound $\symtargetrepairsize{}$, to the entire problem instead of a repair subproblem, by setting $\symfeasibleinteractions{}' = \symfeasibleinteractions{}$.
That yields an algorithm that can find exact solutions and lower bounds that are not necessarily based on mutually exclusive sets,
thus allowing us to find some optimal solutions even if the bound induced by mutually exclusive sets does not match the minimum sample size.

\subsection{Comparison to SampLNS.}
Aside from preprocessing and a new initial heuristic, Sammy differs from SampLNS by using a different approach to find lower bounds.
Although both algorithms mostly rely on mutually exclusive sets as certificates for lower bounds, SampLNS only considers the feasibility of interactions for this purpose.
In other words, SampLNS considers two interactions $\{\ell_1, \ell_2\}, \{\iota_1, \iota_2\}$ to be mutually exclusive if one of $\{\ell_i, \iota_j\}, i, j \in \{1,2\}$ is an infeasible interaction.
In contrast to this, Sammy considers two interactions mutually exclusive if $\UP(\{\ell_1,\ell_2,\iota_1,\iota_2\}) = \bot$.
Because we learn binary clauses for infeasible interactions after the set of feasible interactions is determined,
any pair of interactions that is considered mutually exclusive by SampLNS is also mutually exclusive for Sammy, while the converse does not hold.
Furthermore, unlike SampLNS, which uses an LNS strategy to find better mutually exclusive sets, we employ the cut, price \& round-based approach outlined in \cref{sec:lower-bounds} for that purpose.

Another important difference is the way in which the repair subproblems are handled.
SampLNS considers one repair subproblem at a time using a CP-SAT formulation, which is solved using a solver that itself uses a parallel portfolio of algorithms.
In contrast to that, each LNS worker thread in Sammy creates individual repair subproblems, solving them using single-threaded SAT solvers
using different strategies for both the repair and destroy operations.

\section{Experiments.}
\label{sec:experiments}
In this section, we empirically analyze our algorithm and compare it to the state of the art.
In particular, we design experiments to answer the following research questions.
\begin{description}
  \item[RQ1] What is the impact of simplification on instance size and performance?
  \item[RQ2] How does our initial heuristic compare to the state of the art with regards to scalability and quality?
  \item[RQ3] How much does our LNS approach improve on scalability, solution quality, and bounds compared to SampLNS?
  \item[RQ4] Are there real-world instances with a gap between the maximum mutually exclusive set and the minimum sample size?
\end{description}

We implemented our algorithm; an open-source implementation called \emph{Sammy} is publicly available\footnote{\url{https://doi.org/10.5281/zenodo.17123426}}.
To address RQ1--RQ4, we run Sammy and existing implementations on a diverse set of benchmark instances with a wide range of sizes.
We use the instance set encompassing 47 small to medium-sized instances also used as benchmark in~\cite{10.1145/3712193}, 
but extend it by also adding 8 large real-world instances.
The resulting \emph{benchmark set} of 55 instances contains instances from various domains,
including automotive software, finance, e-commerce, system programs, communication and gaming.
For some further experiments on a superset of \num{1148} real-world instances, see \cref{sec:extra-experments-and-table}.

On each instance from the benchmark set, we run our algorithm, as well as the following algorithms, five times each:
YASA~\cite{KTS+:VaMoS20}, IncLing~\cite{AKT+:GPCE16}, Chvátal~\cite{C:MOR79,JHF:MODELS11} and ICPL~\cite{MHF:SPLC12},
all of which are well-known heuristics for $t$-wise interaction sampling that have been previously evaluated and have publicly available, free implementations.
Furthermore, we also compare our approach to SampLNS with initial solutions from YASA as presented in~\cite{10.1145/3712193}.
Each of the algorithm is given a time limit of \qty{1}{h} and a memory limit of \qty{93}{GiB};
we relax the memory limit for RQ1 on the largest instances, which may require more memory without simplification.


Experiments were run on a machine with an AMD Ryzen 9 7900 CPU and \qty{96}{GiB} of DDR5 RAM with Ubuntu 24.04.2.
Code was written in \CC 17 and com\-piled with clang 18.1.3.
As LP/IP solver, we use Gurobi 12.0.2; as \symsat{} solvers, we use kissat 4.0.3-rc1,
CaDiCaL 2.1.3, cryptominisat 5.11.11 and Lingeling 1.0.0.

\subsection{Impact of Simplification.}
To determine the impact of simplification and answer RQ1, we ran our algorithm with and without simplification.
\Cref{fig:result-simplification-features} shows a summary of the impact of simplification on the benchmark set for several size measures of our instances.
We see a significant reduction of instance complexity across most size metrics;
in particular, the number of interactions that we have to explicitly cover is significantly reduced by both simplification and universe reduction,
in many cases to below \qty{25}{\percent} of the original number; this leads to a significant reduction in memory requirements for the large instances.
As even the time for finding the first solution is typically reduced by simplification (accounting for simplification time),
and only increases mildly in rare cases, for the remaining discussion, we always applied simplification.
\begin{figure}
  \centering
  \includegraphics{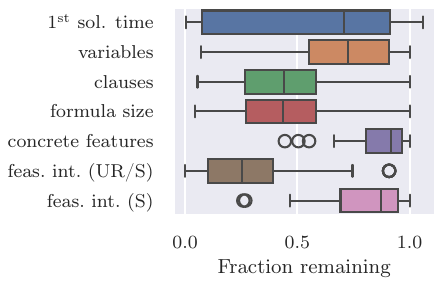}
  \caption{The fraction of each parameter remaining after simplification and universe reduction on the 55 instances from the benchmark set.
           The \emph{time to first solution} bar only includes instances with non-negligible time (at least \qty{0.5}{s}) to find the first solution.
           The \emph{number of feasible interactions} bars show the remaining feasible interactions after just simplification (S) or simplification and universe reduction (UR/S).}
  \label{fig:result-simplification-features}
\end{figure}

\subsection{Runtime and Solution Quality.}
\Cref{fig:result-overview-solution-quality} shows the solution quality 
for the \num{55} instances of the benchmark set, as well as the required runtime.
The solution quality is relative to the best lower bounds achieved by Sammy;
even its worst run on each instance produced a lower bound that was at least as good as the best bound ever reached by SampLNS,
the only other approach capable of producing lower bounds.
In total, we produced better lower bounds than SampLNS on \num{20} of the \num{55} instances;
\cref{tab:instance_summary} shows a table of the instances with the bounds and runtimes achieved by Sammy and SampLNS.

\begin{figure*}
  \centering
  \includegraphics[width=.99\textwidth]{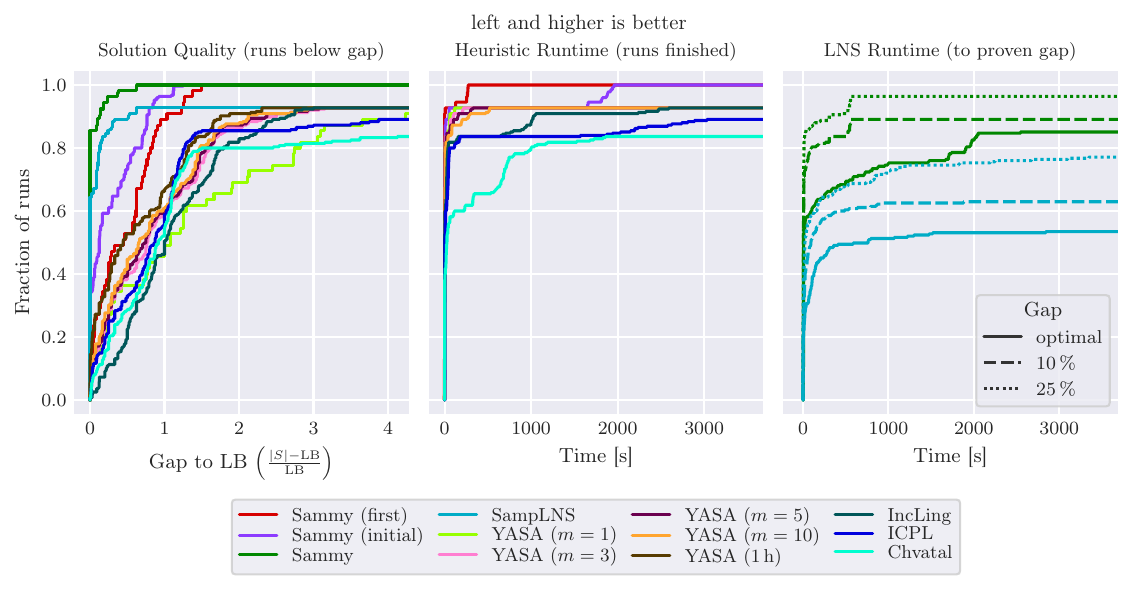}
  \vspace{-1mm}
  \caption{Performance plots showing the solution quality of all algorithms (relative to the best lower bound found by any algorithm),
           the runtime of the initial heuristics and the time taken by the LNS approaches to
           find solutions that they can prove to be within a certain gap to a minimum sample.
           Sammy (first) refers to the first sample found, and Sammy (initial) refers to the entire initial phase.}
  \label{fig:result-overview-solution-quality}
\end{figure*}

To answer RQ2, we observe that the four largest instances in our benchmark set, \texttt{Automotive02\_V0[1-4]}, could only be solved by our implementation;
even with an extended time limit of \qty{3}{h}, previous approaches did not produce a solution for these instances.
Additionally, we see that the average solution quality of our initial solution is better than other approaches that do not require an initial solution to be given,
despite the fact that the other approaches were given up to \qty{1}{h} of runtime,
while our initial heuristic required a maximum of \qty{54}{s} on any instance that was also solved by any other approach.

Regarding RQ3, we find that, aside from solving more instances to provable optimality (\qty{85}{\percent} vs.\ \qty{58}{\percent}) 
and finding better lower bounds for \qty{38}{\percent} of instances, Sammy is also significantly faster than SampLNS.
We also observe that the solutions and bounds we produce are fairly robust across multiple runs of our algorithm.
Only in \num{4} instances from the benchmark set did we observe any deviation in the lower bound value resulting from multiple runs;
the same holds true, but for different instances, for the sample size.
The largest relative gap between two lower bounds was recorded for instance \texttt{Violet} (\num{16} vs.\ \num{17}),
and the largest gap between two sample sizes was recorded for instance \texttt{BattleOfTanks} (\num{283} vs.\ \num{301});
see \cref{tab:instance_summary} for details.

\subsection{Gaps.}
Regarding RQ4, we find \num{16} instances on which our cut, price \& round approach proves that its mutually exclusive set $\symmaxexclusiveset{}$ is maximum,
at least on our subgraph $G_2$, but where we find a provably optimal sample $\symsample{}$ with $|\symsample{}| > |\symmaxexclusiveset{}|$.
In those cases, optimality was proved using one of the \symsat{}-based approaches on the full problem,
either by the exact worker or by eventually removing all configurations in the LNS destroy operation;
the largest gap, both in absolute and relative terms, is a gap between $|\symmaxexclusiveset{}| = 5$ and $|\symsample{}| = 8$ 
for the instances \texttt{FeatureIDE}, \texttt{APL-Model} and \texttt{TightVNC}.

\subsection{Threats to Validity.}
One threat to the validity of findings from empirical analysis of algorithms is the potential for implementation errors.
Fortunately, we have a large variety of existing, well-tested implementations that we can compare our results to.
We verified using mostly independent checking code that each configuration we produced satisfies the original formula $\symfeaturemodel{}$.
We also counted the number of interactions covered for each instance and verify that this number matches between all our simplified and non-simplified runs,
as well as all runs of SampLNS.
Moreover, we verified that each interaction in a mutually exclusive set is among the covered interactions and,
using a simple \symsat{} model, that each reported such set is actually mutually exclusive.

An additional threat arises from the nondeterminism of our algorithm: 
even if we fix all random seeds, the parallelism inherent in our portfolio-based approach introduces nondeterministic behavior;
subproblems considered in an LNS iteration can depend on a race between different solvers on different subproblems in previous iterations.
It is therefore theoretically possible that a lucky or unlucky run is significantly faster or slower or produces significantly better or worse solutions than what we observed in our experiments.
As usual with empirical analysis, despite running our algorithm on a broad set of instances, we also cannot be sure our conclusions generalize to all types of real-world instances.

\section{Conclusion.}
\label{sec:conclusion}
In addition to new theoretical insights into the problem complexity,
we significantly improved the state of the art for solving \pisp{}-instances of all sizes,
considering both lower and upper bounds. Several open questions remain. 
On the theoretical side, it would be interesting to close the remaining gap,
for instance by proving $\P^{\NP[\log]}$-hardness.  On the practical side,
scaling our approaches to $t \geq 3$ and dealing with the huge increase in the
resulting number of interactions is a major open question.
Furthermore, in many applications, the problem of maximizing coverage under
sample size constraints is also of interest, in particular if testing each
configuration incurs considerable costs.

\clearpage
\appendix
\section{Proofs of Complexity Results.}
\label{sec:proofs-complexity}
In this section, we provide the full proofs of the complexity results omitted from \cref{sec:complexity}.

\subsection{Logarithmic Oracle Queries.}
We first prove that logarithmically many queries to a \symsat{} oracle suffice.
\oraclelogqueries*

\begin{proof}
    We establish this theorem by giving an algorithm.
    Its basic idea is to first establish the precise number of
    feasible interactions of the given formula $\symfeaturemodel{}$ and concrete feature set $\symconcretefeatureset{}$ using logarithmically many queries.
    If we know the precise size of the universe,
    a single query to $\symsatoracle{}$ suffices to decide whether
    $\symtargetsamplesize{}$ configurations suffice to cover all feasible interactions.

    The total number of concrete interactions is $M = 2^t{|\symconcretefeatureset{}| \choose t} \in \mathcal{O}(|\symconcretefeatureset{}|^t)$.
    Therefore, the number of feasible concrete interactions is somewhere between $0$ and $M$.
    If we can model the existence of at least $q$ feasible concrete interactions
    as a query $A_1(q)$ to $\symsatoracle{}$, binary search allows us to find the number $M'$
    of feasible interactions using $\mathcal{O}(\log M) = \mathcal{O}(\log |\symconcretefeatureset{}|)$ queries.

    If $M' = 0$, we know the answer is yes.
    Otherwise, a single additional query $A_2(M', \symtargetsamplesize{})$ to $\symsatoracle{}$ then determines whether it is possible to cover at least $M'$ concrete interactions with at most $\symtargetsamplesize{}$ configurations, thus answering the decision problem.

    It remains to show that we can indeed encode $A_1$ and $A_2$ as
    polynomial-size \symsat{} queries.
    $A_1(q)$ can encoded in a formula $\psi_q$ as follows.

    For each of the $M$ possible concrete interactions $\syminteraction{} = \{\ell_{i_1}, \ldots, \ell_{i_t}\}$, we introduce a variable $x_{\syminteraction{}}$.
    Furthermore, for each $\syminteraction{}$, we add to $\psi_q$ a copy of $\symfeaturemodel{}$ on fresh variables, in which the literals from $\syminteraction{}$ are fixed to true, removing satisfied clauses and removing falsified literals from not-yet-satisfied clauses as appropriate.
    We relativize each of the remaining clauses by adding $\overline{x_\syminteraction{}}$ to it;
    observe that $x_\syminteraction{}$ can only be set to true if the copy of $\symfeaturemodel{}$ corresponding to $\syminteraction{}$ is satisfied, i.e., iff $\syminteraction{}$ is a feasible concrete interaction.
    Finally, we add auxiliary variables and clauses that ensure that,
    in any satisfying assignment of the resulting formula, at least $q$ of the variables $x_\syminteraction{}$ are set to true; there are several polynomial-size constructions to achieve this~\cite{bfabc752983f45ce8864a85a1623502b}.

    To encode $A_2(M', \symtargetsamplesize{})$ as a formula $\sigma_\symtargetsamplesize{}$, we begin with $\symtargetsamplesize{}$ copies of $\symfeaturemodel{}$ on fresh variables, allowing $\symsatoracle{}$ to construct $\symtargetsamplesize{}$ independent feasible assignments of $\symfeaturemodel{}$.
    For each interaction $\syminteraction{}$, we introduce variables $x_{\syminteraction{},1}, \ldots, x_{\syminteraction{}, \symtargetsamplesize{}}$, one for each of the $\symtargetsamplesize{}$ copies of $\symfeaturemodel{}$, as well as an additional variable $x_\syminteraction{}$.
    Variable $x_{\syminteraction{},j}$ indicates whether $\syminteraction{}$ is covered by the assignment to the $j$th copy of $\symfeaturemodel{}$, and $x_\syminteraction{}$ indicates whether $\syminteraction{}$ is covered by any assignment.
    To ensure that $x_{\syminteraction{},j}$ takes on the correct value for $\syminteraction{} = \{\ell_1, \ldots, \ell_t\}$, we add clauses $\overline{x_{\syminteraction{},j}} \vee \ell_i^j$ for every $\ell_i \in \syminteraction{}$
    as well as a clause $x_{\syminteraction{},j} \vee \bigvee_{\ell_i \in \syminteraction{}} \overline{\ell_i^j}$;
    here, $\ell_i^j$ denotes the literal corresponding to $\ell_i$ in the $j$th copy of $\symfeaturemodel{}$.
    To ensure that $x_\syminteraction{}$ takes on the correct value, we add clauses $\overline{x_{\syminteraction{},j}} \vee x_\syminteraction{}$ for all $j$ and $\overline{x_\syminteraction{}} \vee \bigvee_{j = 1}^\symtargetsamplesize{} x_{\syminteraction{},j}$.
    This again allows us to add auxiliary variables and clauses to ensure that at least $q$ of the variables $x_\syminteraction{}$ are set to true in any satisfying assignment.
\end{proof}

\subsection{BH-Hardness.}
\bhhardness*

\begin{proof}
Recall that $\BH = \bigcup_{i \in \mathbb{N}} \BH_i$, where $\BH_1 = \NP$, and each subsequent level $\BH_i$ is 
\[\BH_i = \begin{cases}
    \{ A \cap B \mid A \in \coNP, B \in \BH_{i-1}\}, &\text{ if $i$ is even,}\\
    \{ A \cup B \mid A \in \NP, B \in \BH_{i-1}\}, &\text{ if $i$ is odd.}
\end{cases}\]
For each level $k$ of the hierarchy, the alternating satisfiability problem $\text{ASU}_k$ is complete.
  $\text{ASU}_1$ is simply \symsat{}; $\text{ASU}_2$ is also known as \textsc{SatUnsat}.
  Here, the input consists of two formulas $\symredformula{}_1, \symredformula{}_2$ and the problem is to decide whether $\symredformula{}_2$ is UNSAT and $\symredformula{}_1$ is SAT, i.e., whether $\Psi_2 = \text{UNSAT}(\symredformula{}_2) \wedge \text{SAT}(\symredformula{}_1)$ is true.
$\text{ASU}_{k+1}$ extends $\text{ASU}_k$ by adding another formula $\symredformula{}_{k+1}$ to the input and asking whether $\Psi_{k+1} = \text{UNSAT}(\symredformula{}_{k+1}) \wedge \Psi_k$ (if $k+1$ is even) or $\Psi_{k+1} = \text{SAT}(\symredformula{}_{k+1}) \vee \Psi_k$ (if $k+1$ is odd) is true.
  We show that \pisp{} is $\BH$-hard by a reduction from $\text{ASU}_k$ for any $k$.
  Any problem $\mathcal{P} \in \BH$ can then be reduced to \pisp{}:
  \[\mathcal{P} \in \BH \Rightarrow \exists i\ \mathcal{P} \in \BH_i \Rightarrow \exists i\ \mathcal{P} \leq_P^m \text{ASU}_i \leq_P^m \pisp{}.\]

For each formula $\symredformula{}_i(x_1,\ldots,x_{n_i})$ of the $k$ in the input,
we introduce to our resulting formula $\psi$
a set of variables $x^i_{1},\ldots,x^i_{n_i}$,
a variable $x_{\symredformula{}_i}$ that can only be true if $\symredformula{}_i(x^i_1,\ldots,x^i_{n_i})$ is satisfied,
  and a set of concrete features $\symconccomponent{}_i$ of a size $|\symconccomponent{}_i| \geq 4$ that will be specified later.

We enforce that $x_{\symredformula{}_i}$ can only be set to true if $\symredformula{}_i(x^i_1,\ldots,x^i_{n_i})$ is satisfied as follows:
for each clause $\ell_{j_1} \vee \cdots \vee \ell_{j_q}$ in $\symredformula{}_i$, we add the clause $\ell^i_{j_1} \vee \cdots \vee \ell^i_{j_q} \vee \overline{x_{\symredformula{}_i}}$ to $\psi$, where $\ell^i_{j}$ is the literal we obtain by replacing $x_j$ by $x_j^i$ and $\overline{x_j}$ by $\overline{x_j^i}$ in $\ell_j$.
If the formula is satisfied, we allow $x_{\symredformula{}_i}$ to be set to either true or false.

  We add a concrete feature $\symredextraconcrete{}_0$ and, for each formula $\symredformula{}_i$, another (non-concrete) variable $\symredextraconcrete{}_{i}$ and
clauses $\overline{\symredextraconcrete{}_{i}} \vee \overline{\symredextraconcrete{}_{j}}$ for all $i \neq j$ ensuring at most one of the $\symredextraconcrete{}_i$ is set to true.
We also add, for each $c \in \symconccomponent{}_i$, clauses $\symredextraconcrete{}_i \vee \overline{c}$ ensuring that 
concrete features in $\symconccomponent{}_i$ can only be set to true if $\symredextraconcrete{}_i$ is true.
Intuitively speaking, this means that in any configuration, we have to select at most one $i$ for which we are allowed to cover non-false interactions.
Feature $\symredextraconcrete{}_0$ ensures there is one extra configuration covering all interactions in which $\symredextraconcrete{}_0$ is true, and all other concrete features are false.

This has several consequences.
Firstly, true interactions between $c \in \symconccomponent{}_i, d \in \symconccomponent{}_j \neq \symconccomponent{}_i$ become infeasible
and need not be covered.
Secondly, all false interactions within each $\symconccomponent{}_i$ and between different $\symconccomponent{}_i, \symconccomponent{}_j$ are trivially covered in the configuration in which $\symredextraconcrete{}_0$ is true.
Thirdly, if a $c \in \symconccomponent{}_i$ is feasible in any configuration, since $|\symconccomponent{}_i| \geq 4 > 1$, there must be a
configuration in which $c$ is true to cover its interactions with another $d \in \symconccomponent{}_i$.
This configuration automatically covers all mixed interactions of $c$ with $\symredextraconcrete{}_0$ and with any $d \in \symconccomponent{}_j \neq \symconccomponent{}_i$.

Let $\kappa_i$ be the number of configurations required to cover the non-false interactions between the features in $\symconccomponent{}_i$.
The total number of configurations required to cover all interactions in our formula is then $1 + \sum_{i=1}^k \kappa_i$.

Each of the formulas in our input has a (fixed) desired status:
it either occurs in $\Psi_k$ as $\text{SAT}(\symredformula{}_i)$ or $\text{UNSAT}(\symredformula{}_i)$. 
A key idea of our reduction is to make $\kappa_i$ smaller if $\symredformula{}_i$ matches its desired status and larger if it does not.
In particular, if $\symredformula{}_i$ matches its desired status, we have $\kappa_i = \kappa_i^+$, and if it does not, we have $\kappa_i = \kappa_i^-$ with $\kappa_i^+ < \kappa_i^-$.

If $\symredformula{}_i$ occurs as $\text{UNSAT}(\symredformula{}_i)$,
we add the clause $x_{\symredformula{}_i} \vee \overline{c}$ for all $c \in \symconccomponent{}_i$, forcing all $c \in \symconccomponent{}_i$ to false if $x_{\symredformula{}_i}$ is false.
In particular, if $\symredformula{}_i$ is actually UNSAT, this means that no non-false interactions on $\symconccomponent{}_i$ are feasible.
This makes it easy to cover all feasible interactions involving only features in $\symconccomponent{}_i$:
they are covered by any feasible configuration, in particular the one in which $\symredextraconcrete{}_0$ is true.
Thus, we have $\kappa_i^+ = 0$.
If $\symredformula{}_i$ is actually SAT, non-false interactions on $\symconccomponent{}_i$ become feasible.
We then need to introduce configurations covering them, increasing the number of configurations needed.
We also introduce all clauses of the form $\overline{c} \vee \overline{d} \vee \overline{e}$ for all different $c, d, e \in \symconccomponent{}_i$, ensuring that at most two features in $\symconccomponent{}_i$ can simultaneously be true.
This ensures that we need $\kappa_i^- = {|\symconccomponent{}_i|\choose 2}$ configurations to cover all true interactions on $\symconccomponent{}_i$, one for each true interaction.
Because $|\symconccomponent{}_i| \geq 4$, covering each individual true interaction also covers all mixed and false interactions on $\symconccomponent{}_i$.

If $\symredformula{}_i$ occurs as $\text{SAT}(\symredformula{}_i)$, we instead add all clauses of the form $x_{\symredformula{}_i} \vee \overline{c} \vee \overline{d} \vee \overline{e}$ for all different $c, d, e \in \symconccomponent{}_i$.
If $\symredformula{}_i$ is actually SAT, we can set $x_{\symredformula{}_i}$ to true, satisfying all these clauses.
Hence, if we set $\symredextraconcrete{}_i$ and $x_{\symredformula{}_i}$ to true in a configuration, we are free to assign any truth values to the features in $\symconccomponent{}_i$.
In that case, we require $\kappa_i^+ \leq \text{CA}(2, |\symconccomponent{}_i|, 2) = (1 + o(1))\log |\symconccomponent{}_i|$ configurations~\cite{DBLP:journals/siamdm/SarkarC17}, where $\text{CA}(t, c, v)$ is the \emph{covering array number} for $t$-wise interactions on $c$ features with alphabet size $v$.
If $\symredformula{}_i$ is actually UNSAT, the added clauses do not allow us to set three or more features to true simultaneously.
We can thus only cover at most one true interaction on $\symconccomponent{}_i$ in each configuration with $\symredextraconcrete{}_i$ set to true; 
we thus need $\kappa_i^- = {|\symconccomponent{}_i| \choose 2}$ configurations to cover the true (and implicitly, all mixed and false) interactions on $\symconccomponent{}_i$.

By choosing $|\symconccomponent{}_i|$ appropriately, we can create arbitrarily large gaps between $\kappa_i^+$ and $\kappa_i^-$.
Similarly, we can scale the number $\kappa_i$ of configurations required, independently of the formula $\symredformula{}_i$.

We use the above construction to create a \pisp{} instance from an $\text{ASU}_k$ instance recursively as follows.
As a base case, for $k = 2$, i.e., if we want to decide whether $\text{UNSAT}(\symredformula{}_2) \wedge \text{SAT}(\symredformula{}_1)$, we choose $|\symconccomponent{}_1| = |\symconccomponent{}_2| = 4$.

Note that the set of feasible assignments to the concrete features $\{\symredextraconcrete{}_0\} \cup \bigcup_{i=1}^k \symconccomponent{}_i$ does not depend on the formulas $\symredformula{}_i$, but only on their satisfiability.
We can thus enumerate the set of all feasible assignments to the concrete features for any fixed $k$, given a particular satisfiability status of the formulas, and compute the number of required configurations for that status in constant time using a \textsc{Set Cover} solver.

Doing this for $k = 2$ and $|\symconccomponent{}_1| = |\symconccomponent{}_2| = 4$ shows that the resulting instance needs $5$ configurations if $\symredformula{}_2$ is UNSAT and $\symredformula{}_1$ is SAT,
$11$ if both are SAT,
$13$ if $\symredformula{}_2$ is SAT and $\symredformula{}_1$ is UNSAT,
and $7$ if both are UNSAT, 
corresponding to the values $\kappa_2^+ = 0, \kappa_2^- = 6$ and $\kappa_1^+ = 4, \kappa_1^- = 6$.
Thus, we require our \pisp{} instance to have at most $g_{2} = 5$ configurations.

For $k > 2$, if the $\text{ASU}_k$ instance has the form $\text{SAT}(\symredformula{}_k) \vee \Psi$, i.e., $k \geq 3$ and odd, we first construct a \pisp{} instance $\sympispinstance{}_{\Psi}$ for $\Psi$ with some bound $g_{k-1}$ on the number of configurations.
We then extend it by adding $\symredformula{}_k$ as described above.
To do this, we must choose a suitable $|\symconccomponent{}_k|$ and a suitable bound $g_k$ on the number of configurations allowed in a yes-instance.

Note that $g_{k-1}$ only depends on $k$, the level of the boolean hierarchy, and not on any formula.
We also determine the constant $g_{k-1}^- = 1 + \sum_{i=1}^{k-1} \kappa_i^-$, the highest number of configurations that can be required if $\sympispinstance{}_{\Psi}$ is a no-instance.
We must choose $|\symconccomponent{}_k|$ such that the following is ensured. 
If $\symredformula{}_k$ is SAT, the instance should be a yes-instance, forcing us to set $g_k \geq \kappa_k^+ + g_{k-1}^-$.
If $\symredformula{}_k$ is UNSAT, the instance should be a yes-instance iff $\sympispinstance{}_{\Psi}$ is a yes-instance.
In that case, we have to set $g_k = \kappa_k^- + g_{k-1}$; in other words, we have to make $|\symconccomponent{}_k|$ large enough such that $\kappa_k^- + g_{k-1} \geq \kappa_k^+ + g_{k-1}^- \Leftrightarrow \kappa_k^- - \kappa_k^+ \geq g_k^- - g_{k-1}$.
Because $\kappa_k^+$ grows logarithmically and $\kappa_k^-$ grows quadratically with $|\symconccomponent{}_k|$, we can always find a sufficiently large $|\symconccomponent{}_k|$.

If the $\text{ASU}_k$ instance has the form $\text{UNSAT}(\symredformula{}_k) \wedge \Psi$, i.e., if $k \geq 4$ and even, we also first construct a \pisp{} instance $\sympispinstance{}_{\Psi}$ with some bound $g_{k-1}$ on the number of configurations and extend it by adding $\symredformula{}_k$.
We also compute the constant $g_{k-1}^+ = 1 + \sum_{i=1}^{k-1} \kappa_i^+$, the lowest number of configurations required if $\sympispinstance{}_{\Psi}$ is a yes-instance.
If $\symredformula{}_k$ is SAT, the instance needs to be a no-instance.
We thus have to ensure $g_k < \kappa_k^- + g_{k-1}^+$.
If $\symredformula{}_k$ is UNSAT, the instance needs to be a yes-instance iff $\sympispinstance{}_{\Psi}$ is a yes-instance.
We thus have to set $g_k = \kappa_k^+ + g_{k-1} = g_{k-1}$.
Thus, we have to ensure $\kappa_k^- > g_{k-1} - g_{k-1}^+$, which we can always do by choosing a sufficiently large $|\symconccomponent{}_k|$.

While our reduction is non-constructive in the sense that we cannot provide
closed forms for the constants $g_k$ and $\kappa_i^+$ for $\symredformula{}_i$ that occur as $\text{SAT}(\symredformula{}_i)$, for any fixed $k$, the construction can be performed in polynomial time.
\end{proof}

\subsection{Equivalence of Variants.}
\variantsreduction*
We begin by proving the polynomial-time equivalence between \tisp{} and \tispot{}.
\begin{lemma}
  \label{lem:tisp-tispot}
  $\tisp{}$ and $\tispot{}$ are polynomial-time equivalent.
\end{lemma}
We take care not to add non-concrete features in our reduction;
this allows us to use a completely analogous proof for the following lemma.
\begin{lemma}
  \label{lem:tispac-tispacot}
  \tispac{} and \tispacot{} are polynomial-time equivalent.
\end{lemma}
\begin{proof}[Proof of \cref{lem:tisp-tispot}]
  We first show $\tisp{} \leq_m^P \tispot{}$; consider a \tisp-instance $(\symfeatureset{}, \symfeaturemodel{}, \symconcretefeatureset{}, \symtargetsamplesize{})$.
  If $t > |\symconcretefeatureset{}|$, there are no feasible $t$-wise interactions; in that case, our reduction produces some fixed yes-instance of \tispot{}.
  Otherwise, our reduction extends the input instance as follows.
  For each concrete feature $x \in \symconcretefeatureset{}$,
  we add a concrete feature $n_x$ which is forced by clauses $x \vee n_x, \overline{x} \vee \overline{n_x}$
  to take on the value $\overline{x}$ in any satisfying assignment.
  Clearly, we can extend a sample of the original instance by setting $n_x = \overline{x}$,
  obtaining a sample covering all feasible true interactions of the new instance.
  Similarly, we can simply remove the features $n_x$ from a sample of the new instance;
  any feasible interaction of the original instance is covered since it can be translated to a true interaction in the new instance.
  
  To show $\tispot{} \leq_m^P \tisp{}$, we consider a \tispot{}-instance $\mathcal{R} = (\symfeatureset{}, \symfeaturemodel{}, \symconcretefeatureset{}, \symtargetsamplesize{})$.
  Again, we handle the case $t > |\symconcretefeatureset{}|$ by producing a simple yes-instance.

  Otherwise, let $\symsample{}_c$ be the set of all $\sum_{i=0}^{t-1} {|\symconcretefeatureset{}| \choose i}$ assignments on $\symfeatureset{}$
  that make some subset of at most $t-1$ features from $\symconcretefeatureset{}$ true and all other features false.
  For any fixed $t$, $|\symsample{}_c|$ is polynomial in $|\symconcretefeatureset{}|$.
  To reduce our \tispot{}-instance to an instance $\mathcal{T}$ of $\tisp{}$,
  we extend the instance by introducing $t+1$ concrete \emph{cheat features} $x_c^1, \ldots, x_c^{t+1}$
  as well as $(t+1) \cdot |\symsample{}_c|$ concrete \emph{assignment features} $a_1^j, \ldots, a_{|\symsample{}_c|}^j$ for $1 \leq j \leq t+1$.
  We extend each clause $\symclause{}$ of $\symfeaturemodel{}$ to $\symclause{} \vee x_c^1 \vee \cdots \vee x_c^{t+1}$;
  this allows us to ignore the constraints imposed by $\symfeaturemodel{}$ by setting any $x_c^i$ to true.
  We add clauses $\overline{x_c^i} \vee \overline{x_c^j}$ for all $i \neq j$ to ensure
  at most one of the cheat features can be simultaneously true.
  Similarly, for each $a_i^j \neq a_{i'}^{j'}$, we add the clause $\overline{a_i^j} \vee \overline{a_{i'}^{j'}}$.
  We also add clauses $\overline{x_c^j} \vee \bigvee_{i=1}^{|\symsample{}_c|} a_i^j$ for all $1 \leq j \leq t+1$ to ensure that
  setting any cheat feature $x_c^j$ forces us to set a corresponding assignment feature $a_i^j$.
  Similarly, clauses $\overline{a_i^j} \vee x_c^j$ for all $i$ and $j$ ensure that setting any assignment feature requires us to set a corresponding cheat feature.
  Each assignment feature $a_i^j$ has a corresponding assignment with at most $t-1$ true features in $\symconfiguration{}_i \in \symsample{}_c$.
  Clauses $\overline{a_i^j} \vee \ell$ for each $a_i^j$ and each $\ell \in \symconfiguration{}_i$ ensure that setting $a_i^j$ forces us to set all original features to the value they have in $\symconfiguration{}_i$.
  Finally, we obtain our instance $\mathcal{T}$ by requiring at most $\symtargetsamplesize{}(\mathcal{T}) = (t+1) \cdot |\symsample{}_c| + \symtargetsamplesize{}$ configurations in our sample.

  Let $\symsample{}$ be a sample of size at most $\symtargetsamplesize{}$ for $\mathcal{R}$ covering all feasible $t$-wise true interactions on $\symconcretefeatureset{}$.
  We obtain a sample $\symsample{}(\mathcal{T})$ of size $|\symsample{}| + (t+1) \cdot |\symsample{}_c|$ for $\mathcal{T}$ as follows.
  Each configuration $D \in \symsample{}$ is extended by setting all assignment and cheat features to false and then added to $\symsample{}(\mathcal{T})$.
  Then, for each pair $x_c^j$ and $a_i^j$, we add a \emph{cheating configuration} setting these features to true;
  this fixes all other assignment and cheat features to false and all original features according to $\symconfiguration{}_i \in \symsample{}_c$.

  We claim that this sample covers all feasible concrete $t$-wise interactions of $\mathcal{T}$.
  Let $A = \bigcup_{j=1}^{t+1}\{\overline{x_c^j}, \overline{a_1^j}, \ldots, \overline{a_{|\symsample{}_c|}^j}\}$ be the set of negative literals of assignment and cheat features.
  Firstly, all true interactions on $\symconcretefeatureset{}$ are covered by the extended configurations from $\symsample{}$.
  All other $t$-wise interactions on $\symconcretefeatureset{}$ contain at most $t-1$ positive literals;
  all such interactions are covered by some $\symconfiguration{}_i \in \symsample{}_c$, and are thus covered by the cheating configuration induced by $x_c^1$ and $a_i^1$.
  Secondly, there is exactly one satisfying assignment making $x_c^j$ and $a_i^j$ simultaneously true.
  Our sample thus contains \emph{all} satisfying assignments in which a $x_c^j$ or a $a_i^j$ is true;
  therefore, all interactions containing any $x_c^j$ or $a_i^j$ as positive literal are covered.
  Thirdly, each $t$-wise interaction $\syminteraction{} \subset A$ is covered:
  by the pigeon-hole principle, there always is at least one $j$ for which neither $\overline{x_c^j}$ nor any $\overline{a_i^j}$ are in $\syminteraction{}$;
  any cheating configuration in which $x_c^j$ is true covers $\syminteraction{}$.
  It remains to consider $t$-wise interactions $\syminteraction{}$ that contain some feature literals from $\symconcretefeatureset{}$ as well as some from $A$.
  Let $H \subset \syminteraction{}$ be the feature literals in $\syminteraction{}$ over variables from $\symconcretefeatureset{}$.
  Again, by the pigeon-hole principle, there is at least one $j$ such that neither $\overline{x_c^j}$ nor any $\overline{a_i^j}$ are contained in $\syminteraction{}$.
  Furthermore, as $|H| \leq t-1$, there is an $i$ such that the positive literals in $H$ are exactly the positive literals in $\symconfiguration{}_i$.
  Thus, $\syminteraction{}$ is covered by the cheating configuration induced by $x_c^j$ and $a_i^j$.

  Conversely, let $\symsample{}(\mathcal{T})$ be a solution of $\mathcal{T}$ with at most $(t+1)|\symsample{}_c| + s$ configurations.
  By extending the sets $\{x_c^j, a_i^j\}$ to size $t$ with negative literals from $\symconcretefeatureset{}$, we obtain a set of $(t+1)|\symsample{}_c|$ feasible concrete interactions.
  All these interactions are mutually exclusive: $a_i^j$ cannot coexist with any other $a_{i'}^{j'}$ in a configuration.
  Additionally, setting $a_i^j$ already fixes the values of all variables in the entire configuration.
  Therefore, by dropping duplicate configurations, w.l.o.g., assume that $\symsample{}(\mathcal{T})$ contains exactly $(t+1)|\symsample{}_c|$ configurations
  with some cheating feature set to true, one for each $a_i^j$.
  None of these configurations covers any true $t$-wise interaction on $\symconcretefeatureset{}$:
  in any of these configurations, at most $t-1$ variables in $\symconcretefeatureset{}$ are set to true.
  Therefore, after removing the cheating configurations, we end up with at most $\symtargetsamplesize{}$ configurations covering all feasible true interactions on $\symconcretefeatureset{}$.
  By dropping the cheat and assignment features from these configurations, we obtain a solution of $\mathcal{R}$.
\end{proof}

We can now proceed to prove \cref{thm:variants-reduction}.
\begin{proof}[Proof of~\cref{thm:variants-reduction}]
  We prove \[\tisp \equiv_m^P \tispot \leq_m^P \tispacot \equiv_m^P \tispac\text{.}\]
  The two equivalences are established by \cref{lem:tisp-tispot,lem:tispac-tispacot}.
  The result follows because $\tispac \leq_m^P \tisp$ is obvious --- we can just explicitly set $\symconcretefeatureset{} = \symfeatureset{}$.

  It thus remains to prove $\tispot \leq_m^P \tispacot$.
  For this reduction, let $\mathcal{R} = (\symfeatureset{}, \symfeaturemodel{}, \symconcretefeatureset{}, \symtargetsamplesize{})$ be a \tispot-instance.
  We exclude the case $|\symconcretefeatureset{}| < t$ by producing a trivial yes-instance.
  We obtain the $\tispacot$-instance by extending $\mathcal{R}$ as follows.
  For each of the $r = {|\symconcretefeatureset{}| \choose t-1}$ subsets $O_j \subset \symconcretefeatureset{}$ of $t-1$ features from $\symconcretefeatureset{}$,
  we add an \emph{auxiliary} feature $x_c^j$, extending each clause $\symclause{}$ in $\symfeaturemodel{}$ to $\symclause{} \vee \bigvee_{j=1}^t x_c^j$.
  Clauses $\overline{x_c^j} \vee \overline{x_c^{j'}}$ for all $j \neq j'$ ensure at most one auxiliary feature can be made true.
  For each $1 \leq j \leq r$, let $A_j = (\symfeatureset{} \setminus \symconcretefeatureset{}) \cup O_j$ be the set of non-concrete features,
  plus the $t-1$ concrete features from $O_j$.
  We add binary clauses enforcing that setting $x_c^j$ to true forces the features in $A_j$ to true and all others to false.
  Finally, we obtain our \tispacot{}-instance by requiring at most $\symtargetsamplesize{}(\mathcal{T}) = r + \symtargetsamplesize{}$ configurations.

  Let $\symsample{}$ be a solution of size at most $\symtargetsamplesize{}$ of $\mathcal{R}$.
  We obtain a solution of $\symsample{}(\mathcal{T})$ as follows:
  we extend each configuration in $\symsample{}$ by setting all auxiliary features to false.
  Then, we add a single configuration for each of the $r$ auxiliary features $x_c^j$,
  in which precisely that feature and all features from $A_j$ are set to true.
  We claim that this covers all feasible true $t$-wise interactions on $\mathcal{T}$.
  Because $\symsample{}$ covers all feasible true $t$-wise interactions on $\symconcretefeatureset{}$, these are also covered in $\symsample{}(\mathcal{T})$.
  For any $x_c^j$, there is exactly one satisfying assignment where it is true, and that assignment is part of $\symsample{}(\mathcal{T})$.
  Thus, all feasible true interactions involving any $x_c^j$ are also covered.
  Finally, let $\syminteraction{}$ be a true $t$-wise interaction involving at least one feature from $\symfeatureset{} \setminus \symconcretefeatureset{}$ and some (potentially empty) set of features from $\symconcretefeatureset{}$.
  Then, $\syminteraction{}$ is also covered: $\syminteraction{}$ contains at most $t-1$ features from $\symconcretefeatureset{}$, 
  so this set of features is contained in at least one $A_j$ and thus covered by the configuration in which $x_c^j$ is true.

  Conversely, let $\symsample{}$ be a solution of size at most $r + \symtargetsamplesize{}$ of $\mathcal{T}$.
  Observe that $\symsample{}$ must contain at least $r$ configurations in which some $x_c^j$ is set to true
  to cover true $t$-wise interactions containing $x_c^j$;
  by dropping duplicate configurations, w.l.o.g., we can assume that $\symsample{}$ contains exactly $r$ such configurations.
  We further observe that none of these configurations covers any true $t$-wise configuration on $\symconcretefeatureset{}$:
  in all of these configurations, at most $t-1$ features from $\symconcretefeatureset{}$ are simultaneously true.
  Thus, any feasible true $t$-wise interaction must be covered by the at most $\symtargetsamplesize{}$ remaining configurations in which all $x_c^j$ are false.
  Therefore, dropping $x_c^j$ from these remaining configurations yields a solution of size at most $\symtargetsamplesize{}$ for $\mathcal{R}$.
\end{proof}

\subsection{Hardness for Larger $t$.}
\hardnesslargert*
\begin{proof}
  For $t' = t$, there is nothing to prove.
  Otherwise, let $\Delta_t = t - t'$ and $\mathcal{R} = (\symfeatureset{}, \symfeaturemodel{}, \symconcretefeatureset{}, \symtargetsamplesize{})$ be a $t'\textsc{-Isp-Ot}$-instance.
  We assume $\Delta_t \geq |\symconcretefeatureset{}|$; else we construct a simple yes-instance.
  Our reduction extends this to a \tispot-instance $\mathcal{T}$ as follows.
  Firstly, we add an auxiliary non-concrete feature $x_c$ and $\Delta_t$ concrete features $a_1,\ldots,a_{\Delta_t}$,
  as well as auxiliary concrete features $y_1,\ldots,y_{\Delta_t}$.
  Using clauses $\overline{x_c} \vee \bigvee_{i=1}^{\Delta_t} a_i$, $\overline{a_i} \vee x_c$ and $\overline{a_i} \vee \overline{a_j}$ for all $i \neq j$,
  we enforce that we can only make $x_c$ true if we make exactly one $a_i$ true.
  Moreover, any $a_i$ being true forces $x_c$ to true and all other $a_i$ to false.
  Furthermore, we extend all clauses $\symclause{}$ of $\symfeaturemodel{}$ to $\symclause{} \vee x_c$.
  We also add clauses $\overline{a_i} \vee \overline{y_i}$ for all $i$, $\overline{a_i} \vee y_j$ for all $i \neq j$,
  and $\overline{a_i} \vee x$ for all $x \in \symfeatureset{}$.
  In other words, setting $a_i$ to true forces $y_i$ to false and all other $y_j$ as well as all original features to true.
  We obtain our \tispot-instance by allowing at most $\symtargetsamplesize{}(\mathcal{T}) = \symtargetsamplesize{} + \Delta_t$ configurations.
  
  Let $\symsample{}$ be a solution to $\mathcal{R}$ of size at most $\symtargetsamplesize{}$.
  We obtain a solution $\symsample{}(\mathcal{T})$ to $\mathcal{T}$ as follows.
  We extend each configuration in $\symsample{}$ by setting $x_c$ and all $a_i$ to false and all $y_i$ to true.
  In addition, for each $a_i$, we introduce one configuration in which $x_c$ and that $a_i$ are set to true.

  We claim that this covers all feasible concrete true $t$-wise interactions.
  Let $Y = \{y_1, \ldots, y_{\Delta_t}\}$ and $A = \{a_1, \ldots, a_{\Delta_t}\}$;
  the concrete features of $\mathcal{T}$ are $\symconcretefeatureset{}' = \symconcretefeatureset{} \cup Y \cup A$.
  For each $a_i$, there is exactly one satisfying assignment in which $a_i$ is true;
  that assignment is part of our solution $\symsample{}(\mathcal{T})$.
  Therefore, all feasible concrete true $t$-wise interactions $\syminteraction{}$ with $\syminteraction{} \cap A \neq \emptyset$ are covered.
  Let $Y' \subsetneq Y$ and consider a concrete interaction $\syminteraction{} = Y' \cup \Gamma$ for some $\Gamma \subset \symconcretefeatureset{}$.
  Let $y_j \in Y \setminus Y'$ arbitrary; then, $\syminteraction{}$ is covered by the configuration with $x_c$ and $a_j$ set to true.
  Now, consider a concrete interaction $\syminteraction{} = Y \cup \Gamma$ for some $\Gamma \subset \symconcretefeatureset{}$.
  Then, $\Gamma$ is a concrete true $t'$-wise interaction in $\mathcal{R}$; 
  it therefore is either covered in or infeasible for $\mathcal{R}$.
  If it is covered in some configuration $D \in \symsample{}$, the corresponding configuration in $\symsample{}(\mathcal{T})$ covers $\syminteraction{}$ in $\mathcal{T}$.
  If it is infeasible for $\mathcal{R}$, then $\syminteraction{}$ is infeasible for $\mathcal{T}$:
  because $Y \subset \syminteraction{}$, all $y_i$ must be true; that is only possible if $x_c$ is false,
  in which case we have to assign values to $\symfeatureset{}$ such that the clauses of $\symfeaturemodel{}$ are satisfied.

  Conversely, let $\symsample{}(\mathcal{T})$ be a solution to $\mathcal{T}$ of size at most $\symtargetsamplesize{} + \Delta_t$.
  We observe that, for each $a_i$, there are are feasible concrete $t$-wise true interactions containing $a_i$;
  therefore, $\symsample{}(\mathcal{T})$ contains at least one configuration in which that $a_i$, and thus $x_c$, is true.
  We drop all configurations in which $x_c$ is true and remove the additional features from the remaining configurations to obtain a solution $\symsample{}$ for $\mathcal{R}$ with $|\symsample{}| \leq \symtargetsamplesize{}$.
  We claim that solution covers all feasible concrete true $t'$-wise interactions of $\mathcal{R}$.
  Let $\syminteraction{}$ be such an interaction and $D_\syminteraction{}$ a configuration covering it.
  Then we consider $\syminteraction{}' = Y \cup \syminteraction{}$.
  $\syminteraction{}'$ is clearly feasible for $\mathcal{T}$: we can extend $D_\syminteraction{}$ by setting $x_c$ to true, all $a_i$ to false and all $y_i$ to true.
  Thus, there must be a configuration $D_{\syminteraction{}'} \in \symsample{}(\mathcal{T})$ covering $\syminteraction{}'$.
  This configuration contains all of $Y$ and thus $\overline{x_c}$.
  Its restriction to $\symfeatureset{}$ is thus part of $\symsample{}$; therefore, $\syminteraction{}$ is covered by $\symsample{}$.
  This concludes the proof.
\end{proof}

\section{Details for Initial Heuristic.}
\label{sec:details-heuristic}
In this section, we describe some implementation details left out of \cref{sec:initial-heuristic} due to space constraints.

\subsection{Trail Data Structure.}
Each partial configuration is stored in a trail data structure,
which consists of a stack of literals and their reasons, 
as well as a data structure that tracks the status of each variable and its trail position.

Some of the large instances require a large number of configurations.
In some large instances, we have to support up to \num{50000} trails simultaneously.
Compared to a regular CDCL SAT solver, which usually only has to keep one trail --- or a few, e.g., for parallel solvers --- this necessitates some changes.
For instance, to conserve memory, it is imperative that we do not store a full copy of the entire clause set for each trail, but share the clause database among all trails.
This disables some standard optimizations, e.g., reordering of the literals in each longer clause to track watched literals;
each trail thus needs to track which literals are watched in each clause separately.

Additionally, for each literal, we maintain a list of clauses in which it is watched;
most of the time, these lists remain empty or very short.
Using trivial dynamic arrays for these lists cause a large number of small memory allocations;
we instead use dynamic arrays with a small amount of static storage to avoid most of these. 
However, a significant amount of memory is still reserved in small allocations by our initial heuristic.

With some memory allocators, such as the default \texttt{malloc} provided by glibc on our Linux system, this causes significant issues:
while our initial heuristic requires significant memory in small allocations, our LNS approach almost exclusively allocates large chunks of memory.
These are handled separately; as is turns out, \texttt{malloc} never returns the space of the small allocations to the operating system when they are \texttt{free}d,
and it also fails to reuse them to satisfy the requests for large chunks of memory later in our algorithm.
Without a call to \texttt{malloc\_trim}, which causes \texttt{malloc} to actually return memory to the OS after our initial phase,
our algorithm thus runs out of memory on the largest instances of the benchmark set.

\subsection{Tracking Coverage.}
The high-level idea of our algorithm requires us to enumerate or sample from the set of uncovered, potentially feasible interactions.
This could theoretically be achieved by tracking, at any point in our algorithm, which interactions are covered.
For large instances\footnote{Some of our instances have more than \num{500000000} feasible interactions and yield samples of \num{40000} configurations after the first iteration.}, 
in particular those that require a large sample to achieve pairwise coverage, this task is not trivial.
During preliminary experiments, we attempted to use a simple matrix data structure counting, for each pair of literals, the number of partial configurations covering it.
Profiling indicated that more than \qty{99.9}{\percent} of our runtime was spent updating this matrix;
note that, by performing a single \emph{push\_and\_propagate} operation on some partial configuration $\sympartialconfig{}$,
we may have to consider and add coverage for $\omega(n)$ interactions with potentially poor memory locality,
because we may end up adding more than $O(1)$ concrete feature literals due to UP.

We thus designed a different approach: for each configuration $\sympartialconfig{}$, we maintain a bitset $B_{t}(\sympartialconfig{})$ of the $2n$ literals, indicating which of the literals are true in $\sympartialconfig{}$.
We also maintain a bitset $B_{x}(\ell)$ for each literal $\ell$, indicating for which literals $\ell'$ we have established that $\{\ell, \ell'\}$ is an infeasible interaction.
We also maintain, for each literal $\ell$, a list of class indices in which $\ell$ is true.
These sets are much cheaper to maintain on changes to $\sympartialconfig{}$.
However, they do make other operations more expensive; this primarily affects enumerating currently uncovered interactions.

\subsection{Enumerating Uncovered Interactions.}
To iterate uncovered interactions, we iterate through all (concrete) feature literals $\ell$.
For each $\ell$, we combine $B_{x}(\ell)$ and $B_t(\sympartialconfig{})$ for all $\sympartialconfig{}$ in which $\ell$ is true,
resulting in another bitset indicating which interactions involving $\ell$ are currently covered, allowing us to iterate the uncovered interactions involving $\ell$.
If we have not established which interactions are actually infeasible, unidentified infeasible interactions are enumerated as well.
Even though bitwise logic operations can be performed extremely efficiently, scanning linearly through memory and manipulating hundreds or thousands of bits each clock cycle as long as memory bandwidth permits,
this is a relatively expensive operation for large instances that require many configurations.
We thus need to ensure that iterating the set of uncovered interactions is done relatively rarely.
We achieve this by our priority queue $\symqueue{}$ and exponentially growing $\symtargetqueuesize{}$, since we only need to enumerate uncovered interactions if $\symqueue{}$ was emptied.

\subsection{Priority Queue.}
For any interaction $\syminteraction{} = \{\ell_1, \ell_2\}$ that is currently in the priority queue $\symqueue{}$, we maintain the following additional information:
a bitset of partial configuration indices that this interaction could potentially be added to, i.e., partial configurations $\sympartialconfig{}$ with $\overline{\ell_1}, \overline{\ell_2} \notin \sympartialconfig{}$, 
as well as a count $c_\syminteraction{}$ of such partial configurations; this information is what we use to prioritize interactions, starting with interactions with the lowest $c_\syminteraction{}$.

\subsection{Infeasibility Detection.}
If we are unable to add an interaction $\syminteraction{}$ taken from $\symqueue{}$ to any existing partial configuration, we create a new partial configuration $T$ for $\syminteraction{}$. 
In addition to checking for conflicts after pushing $\syminteraction{}$, we can optionally perform some amount of CDCL search on $T$.
This can serve multiple purposes: if we find $\syminteraction{}$ to be infeasible, we can ignore it in the future;
we can also learn clauses from conflicts that strengthen propagation in the future, or find that $\syminteraction{}$ is definitely feasible.
During preliminary experiments, we found that some instances profited from stronger checks of feasibility for such $\syminteraction{}$; 
our current implementation performs a full CDCL search for each such $\syminteraction{}$, either proving $\syminteraction{}$ to be feasible or infeasible.

\subsection{Finalizing.}
When enumerating uncovered interactions detects all interactions have been covered, we are left with the task of completing all our partial configurations $\sympartialconfig{} \in \symworkingsample{}$ to complete configurations.
We extend each configuration $\sympartialconfig{}$ to a complete configuration $\symconfiguration{}$ individually, again using a simplistic CDCL SAT solver based on our trail data structure.
During the operation of this solver, we attempt to maintain the partial configuration $\sympartialconfig{} \subseteq \symconfiguration{}$ as follows.
Whenever a decision from the original $\sympartialconfig{}$ is removed by conflict resolution, we record it as \emph{failed}; 
before making any other decisions in our SAT solver, we attempt to reintroduce failed decisions that are currently open.
Unless the formula is unsatisfiable, this completion routine always produces a valid configuration $\symconfiguration{}$;
however, it may happen that $\sympartialconfig{} \nsubseteq \symconfiguration{}$.
We replace $\sympartialconfig{}$ by $\symconfiguration{}$ in any case; if $\sympartialconfig{} \nsubseteq \symconfiguration{}$ for any $\sympartialconfig{} \in \symworkingsample{}$, 
we need to go back to enumerating uncovered interactions and potentially introduce new partial configurations to cover interactions that became uncovered while completing partial configurations.

\section{Preprocessing.}
\label{sec:sat-preprocessing}
In this section, we argue why our preprocessing is safe.
Essentially, we have to guarantee that we can turn any sample $S_\symsimplified{}$ of our simplified formula $\symsimplified{}$ into a sample $S_\symfeaturemodel{}$ of the original formula $\symfeaturemodel{}$ with $|S_\symsimplified{}| = |S_\symfeaturemodel{}|$ and vice versa, maintaining pairwise coverage.
This guarantees that a minimum sample of $\symfeaturemodel{}$ corresponds to a minimum sample of $\symsimplified{}$ and vice versa;
if a preprocessing rule satisfies this condition, we call it \emph{sampling-safe}.

There is a notable special case to handle that we exclude in the following:
if preprocessing reduces the number of concrete features below $2$,
even otherwise sampling-safe preprocessing rules can violate the above rule because an empty universe of interactions between concrete features allows an empty sample to have pairwise coverage.
However, we can fix this special case by requiring at least one configuration in any sample if the formula is satisfiable, and
requiring two configurations if there is exactly one concrete feature $x$ and both $x$ and $\overline{x}$ are feasible, one with $x$ and one with $\overline{x}$.

\subsection{Failed and Equivalent Literals.}
A literal $\ell$ is called \emph{failed literal} if $\UP(\ell) = \bot$.
This means that $\ell$ must be false in all satisfying assignments,
and we can simplify $\symfeaturemodel{}$ by removing $\ell$, removing all clauses satisfied by setting $\ell$ to false, and shortening all clauses containing $\ell$.
The same holds true if, for any literal $\ell'$, we have $\overline{\ell} \in \UP(\ell')$ and $\overline{\ell} \in \UP(\overline{\ell'})$.

A pair of literals $\ell_1, \ell_2$ are called \emph{equivalent} iff the value of $\ell_1$ and $\ell_2$ are the same in all satisfying assignments.
We can compute a directed graph on all literals with a directed edge $(\ell_1, \ell_2)$ between two literals if $\ell_2 \in \UP(\ell_1)$, i.e., if $\ell_1$ implies $\ell_2$ in all satisfying assignments.
Strongly connected components (SCCs) in this graph represent sets of equivalent literals; if any SCC contains both $\ell$ and $\overline{\ell}$, the formula is unsatisfiable.
We can replace all variables occurring in an SCC by a single variable,
rewriting all clauses containing the involved variables.
If any of the involved variables corresponds to a concrete feature, the resulting variable in the simplified formula is also marked as concrete feature.

\begin{theorem}
Failed and equivalent literal elimination are sampling-safe.
\end{theorem}
\begin{proof}
Let $\symfeaturemodel{}$ be the formula before failed and equivalent literal elimination and $\symsimplified{}$ be the formula afterwards.
Let $S_{\symfeaturemodel{}}$ be a sample with pairwise coverage on $\symfeaturemodel{}$.
If $S_{\symfeaturemodel{}}$ is empty, due to our conventions of handling the cases of zero and one concrete feature, the original formula is unsatisfiable, and so is $\symsimplified{}$.
Otherwise, because every configuration $\symconfiguration{}_{\symfeaturemodel{}} \in S_{\symfeaturemodel{}}$ is a satisfying assignment, all failed literals are false in $\symconfiguration{}_{\symfeaturemodel{}}$.
If two literals are equivalent, they have the same value in $\symconfiguration{}_{\symfeaturemodel{}}$.
Therefore, we can transform each $\symconfiguration{}_{\symfeaturemodel{}}$ into a satisfying assignment $\symconfiguration{}_{\symsimplified{}}$ by dropping variables corresponding to failed literals and replacing equivalent literals.
The sample obtained in this way has pairwise coverage on $\symsimplified{}$.

Let $S_{\symsimplified{}}$ be a sample with pairwise coverage on $\symsimplified{}$.
Again, $S_{\symsimplified{}}$ is only empty if $\symsimplified{}$ and $\symfeaturemodel{}$ are unsatisfiable.
Otherwise, we turn each configuration $\symconfiguration{}_{\symsimplified{}} \in S_{\symsimplified{}}$ into a configuration $\symconfiguration{}_{\symfeaturemodel{}}$ on $\symfeaturemodel{}$ by setting failed literals to false and setting equivalent literals to the values indicated by their representative literal in $\symconfiguration{}_{\symsimplified{}}$, thus obtaining a sample $S_\symfeaturemodel{}$.

  Let $\syminteraction{} = \{\ell_1,\ell_2\}$ be a feasible interaction between concrete literals of $\symfeaturemodel{}$;
we have to show that it is covered in $S_\symfeaturemodel{}$.
If neither $\ell_1$ nor $\ell_2$ were removed by failed or equivalent literal elimination,
we have $\syminteraction{} \subseteq \symconfiguration{}_\symsimplified{}$ for some $\symconfiguration{}_\symsimplified{} \in S_\symsimplified{}$, and the corresponding $\symconfiguration{}_\symfeaturemodel{} \in S_\symfeaturemodel{}$ covers $\syminteraction{}$.

If one of $\ell_1, \ell_2$, w.l.o.g. $\ell_1$, is a negated failed literal, $\ell_1$ is true in all $\symconfiguration{}_\symfeaturemodel{}$.
We thus have to prove that $\ell_2$ is true in some $\symconfiguration{}_\symfeaturemodel{}$.
If $\ell_2$ is also a negated failed literal, this holds for any $\symconfiguration{}_\symfeaturemodel{}$.
If $\ell_2$ is an equivalent literal, let $\ell_3$ be its representative in $\symsimplified{}$;
otherwise, let $\ell_3 = \ell_2$.
It suffices to show that $\ell_3$ is true in some $\symconfiguration{}_\symsimplified{}$.
If $\ell_3$ is the only concrete literal in $\symsimplified{}$, $S_\symsimplified{}$ either has two configurations if both $\ell_3$ and $\overline{\ell_3}$ are feasible, or one configuration if only $\ell_3$ is;
in either case, $\ell_3$ is contained in a $\symconfiguration{}_\symsimplified{}$.
Otherwise, there is another concrete literal $\ell_4$ such that $\{\ell_3, \ell_4\}$ is a feasible interaction of $\symsimplified{}$; because $S_\symsimplified{}$ has pairwise coverage, there must be a configuration $\symconfiguration{}_\symsimplified{}$ in which $\ell_3$ is true.

Finally, let one or both of $\ell_1, \ell_2$ be literals replaced by representatives $\ell_1', \ell_2'$ in $\symsimplified{}$ through equivalent literal elimination.
Because $\ell_1$ and $\ell_2$ are concrete, their representatives are also concrete in $\symsimplified{}$.
Thus, because $S_\symsimplified{}$ has pairwise coverage, there must be a $\symconfiguration{}_\symsimplified{}$ in which $\ell_1'$ and $\ell_2'$ are simultaneously true.
The corresponding $\symconfiguration{}_\symfeaturemodel{}$ covers $\syminteraction{}$.
\end{proof}

\subsection{Bounded Variable Elimination.}
Another important, well-known SAT simplification method is \emph{bounded variable elimination} (BVE).
BVE is based on the observation that one can eliminate any variable $x$ from a formula $\symfeaturemodel{}$ by resolving each clause containing $x$ with each clause containing $\overline{x}$ on $x$, adding all non-tautological resolvents and then removing all clauses with $x$ or $\overline{x}$ as well as variable $x$; the resulting formula $\symsimplified{}$ is satisfiable iff $\symfeaturemodel{}$ was.
In the worst case, this can drastically increase the formula size;
BVE is usually only applied to variables that do not cause a notable increase in formula size.
Since BVE does not result in a logically equivalent formula, we have to be careful when applying BVE; however, the following holds.
\begin{theorem}
BVE is sampling-safe when applied to non-concrete features.
\end{theorem}
\begin{proof}
Let $y$ be a non-concrete feature of $\symfeaturemodel{}$, and let $\symsimplified{}$ be the result of performing BVE on $y$.
  From a sample $S_{\symfeaturemodel{}}$ of $\symfeaturemodel{}$, we obtain a sample $S_{\symsimplified{}}$ simply by dropping $y$ from all configurations in $S_{\symfeaturemodel{}}$; because the set of concrete interactions did not change, $S_{\symsimplified{}}$ has pairwise coverage if $S_{\symfeaturemodel{}}$ has.
  On the other hand, to obtain a sample $S_{\symfeaturemodel{}}$ of $\symfeaturemodel{}$ from a sample $S_{\symsimplified{}}$ of $\symsimplified{}$,
we reintroduce $y$ assigned to some value into each configuration $\symconfiguration{}_{\symsimplified{}} \in S_{\symsimplified{}}$.
Let $\symconfiguration{}_{\symsimplified{}}^+$ be the configuration obtained by adding $y$, and $\symconfiguration{}_{\symsimplified{}}^-$ be the configuration obtained by adding $\overline{y}$ to $\symconfiguration{}_\symsimplified{}$.
If $\symconfiguration{}_\symsimplified{}^+$ is a satisfying assignment of $\symfeaturemodel{}$, we have translated $\symconfiguration{}_\symsimplified{}$ to a satisfying assignment of $\symfeaturemodel{}$ covering the same concrete interactions, and we are done.
If $\symconfiguration{}_\symsimplified{}^+$ is not a satisfying assignment of $\symconfiguration{}_\symfeaturemodel{}$,
this is due to a clause $\symclause{} \in \symfeaturemodel{}$ containing $\overline{y}$.
The resolvent of $\symclause{}$ with every clause $\symclause{}^+$ containing $y$ is part of $\symsimplified{}$ and thus satisfied.
Because $\symclause{}$ is not satisfied by $\symconfiguration{}_\symsimplified{}^+$, this means that every clause $\symclause{}^+$ is satisfied in $\symconfiguration{}_\symsimplified{}^-$ by a literal that is not $y$; therefore $\symconfiguration{}_\symsimplified{}^-$ must be a satisfying assignment of $\symfeaturemodel{}$.
\end{proof}

\subsection{Universe Reduction.}
\label{sec:universe-reduction}
We apply universe reduction after our initial phase, when the set of feasible interactions is known.
Recall that \emph{universe reduction} describes the process of eliminating feasible interactions $\syminteraction{}$ from $\symfeasibleinteractions{}$ that are implicitly covered whenever another, uneliminated interaction $\syminteraction{}'$ is covered.
Our implementation uses two rules, called (I) and (II) in the following, to find such interactions.

Rule (I) is based on the observation that, if $\ell_2 \in \UP(\{\ell_1\})$, then for any $\ell_3$ with $\{\ell_1,\ell_3\} \in \symfeasibleinteractions{}$, $\{\ell_2,\ell_3\}$ is implied by $\{\ell_1, \ell_3\}$.
This rule can be applied quickly by performing UP on each potential $\ell_1$, 
followed by simultaneously walking two sorted lists of potential partner literals $\ell_3$ of $\ell_1$ and $\ell_2$.

Rule (II) is based on UP of interactions: if $\{\ell_3, \ell_4\} \in \UP(\{\ell_1, \ell_2\})$, then $\{\ell_3, \ell_4\}$ is implied by $\{\ell_1, \ell_2\}$.
This requires us to perform UP on interactions.
While this dominates rule (I), it can be quite expensive to perform for large $\symfeasibleinteractions{}$;
therefore, we first perform rule (I) in all cases and only run rule (II) up to a time limit on the remaining interactions.

After this process, each implied interaction has a stored \emph{implier}.
We replace implied impliers by their impliers until each implied interaction has a non-implied implier.
These impliers can be used to remove implied elements from mutually exclusive sets used as lower bounds,
simply by replacing each implied interaction by its implier.

\section{LNS Portfolio Details.}
\label{sec:a-lns}
In this section, we give some more details regarding our main LNS algorithm.

\subsection{Destroy Size.}
The destroy size $\symdestroyparameters{}$, i.e., the number of configurations removed by the destroy operations, 
is governed by the success and performance of previous destroy and repair operations.
We initially choose a deliberately low $\symdestroyparameters{}$ depending on $|\symfeasibleinteractions{}|$.
As long as the previous destroy size appears to be successful, i.e., has led to an improvement in the previous iteration, we do not increase it.
Otherwise, we increment $\symdestroyparameters{}$ after a certain number of destroy and repair operations that did not improve the solution.
This number of unsuccessful operations depends on the average runtime of the repair operations at the current $\symdestroyparameters{}$:
as long as the average repair operation time is below a given bound, we increment $\symdestroyparameters{}$ faster.

\subsection{Removed Configuration Selection.}
Another important parameter concerns the procedure that determines which configurations to remove in the destroy operation.
We consider three different approaches and randomly decide between them.

\subsubsection{Uniformly Random.}
One option is to select the requested number of configurations uniformly at random; we pick this option with a probability of $0.2$.
This works well in some cases, but for some instances runs into \emph{doomed destructions} too frequently.
We call the removal of configurations $\symdestruction{}$ from a sample $\symsample{}$ \emph{doomed} if there is a mutually exclusive set $\symmaxexclusiveset{}$ of uncovered interactions that contains, 
for each $\symconfiguration{} \in \symdestruction{}$, one distinct interaction $\syminteraction{} \subseteq \symconfiguration{}$; such destructions cannot lead to an improvement.
We observe the following.
\begin{observation}
    To find a clique that certifies a destruction $\symdestruction{}$ of a sample $\symsample{}$ to be doomed, 
    it is sufficient to consider only interactions that are 
    covered by exactly one configuration $\symconfiguration{} \in \symsample{}$.
\end{observation}
\begin{proof}
    Let $\symmaxexclusiveset{}$ be a clique that certifies $\symdestruction{}$ to be doomed,
    and let $\syminteraction{} \in \symmaxexclusiveset{}$ be an interaction that is covered by at least 
    two configurations $\symconfiguration{}_1, \symconfiguration{}_2 \in \symsample{}$.
    If $\{\symconfiguration{}_1, \symconfiguration{}_2\} \nsubseteq \symdestruction{}$, then $\syminteraction{}$ is not uncovered by removing $\symdestruction{}$ from $\symsample{}$.
    Otherwise, we obtain a contradiction to $|\symmaxexclusiveset{}| \geq |\symdestruction{}|$, since each configuration in $\symdestruction{}$ can contain at most one interaction $\syminteraction{} \in \symmaxexclusiveset{}$.
\end{proof}

\subsubsection{Avoiding Doomed Destructions.}
Given a table of cliques, which we build from global lower bounds and previous destroy and repair operations, 
we attempt to avoid doomed destructions for a given destroy size $\symdestroyparameters{}$ as follows.
After selecting $d_1 < \symdestroyparameters{}$ configurations $\symdestruction{} \subset \symsample{}$ to remove at random, we consider each clique $\symmaxexclusiveset{}$ in our table.
If at least one of the removed configurations does not contain any interaction from $\symmaxexclusiveset{}$, we can safely ignore $\symmaxexclusiveset{}$.
Otherwise, we enumerate all configurations from $\symsample{}$ that do not contain an interaction from $\symmaxexclusiveset{}$ and randomly extend $\symdestruction{}$ by one of them, 
until all cliques are processed or destroy size $\symdestroyparameters{}$ is reached.
We use this approach with a probability of $0.3$.

\subsubsection{Randomized Greedy.}
Another idea is to mix random destruction with a greedy aspect; we pick this approach with a probability of $0.5$.
For some given destruction size $\symdestroyparameters{}$, we first select $d_1 < \symdestroyparameters{}$ configurations to be removed uniformly at random.
We then determine, for each remaining configuration $\symconfiguration{}_i$, the number of interactions only covered by $\symconfiguration{}_i$ and remove the $\symdestroyparameters{} - d_1$ configurations with the lowest number of uniquely covered interactions.

\subsection{Repair Strategy Selection.}
We currently select the repair strategy randomly between the available approaches,
with a slightly higher chance of using the non-incremental \symsat{} approach.
The non-incremental approach can use one of four \symsat{} solvers (kissat, CaDiCaL, Lingeling and Cryptominisat),
whereas the incremental approach can only use the latter three solvers which support incremental solving.

\section{Extra Experiments and Tables.}
\label{sec:extra-experments-and-table}
For reference purposes, this section contains extra tables of data collected in our main experiment.

\begin{table*}
\sisetup{
exponent-mode = fixed,
fixed-exponent = 0,
}
\centering
  \vspace{-0.3cm}
\begin{tiny}
\resizebox{.75\textwidth}{!}{\rotatebox{90}{%
\begin{tabular}{lrrrrrrrrrr}
Instance & $|\mathcal{F}|$ & $|\mathcal{C}|$ & Clauses & $|\mathcal{I}|$  & LB        & UB    & LB      & UB       & Time  & Time\\
         &                 &                 &         & (frac.\ rem.)    & Sammy     & Sammy & SampLNS & SampLNS  & Sammy & SampLNS\\
\hline
    
lcm & \num[text-series-to-math=true]{7} & \num[text-series-to-math=true]{5} & \num[text-series-to-math=true]{12} & \num[text-series-to-math=true]{37} (\num[text-series-to-math=true]{0.43}) & \textbf{\num[text-series-to-math=true]{6}} & \textbf{\num[text-series-to-math=true]{6}} & \textbf{\num[text-series-to-math=true]{6}} & \textbf{\num[text-series-to-math=true]{6}} & 5.0 & 0.4\\
calculate & \num[text-series-to-math=true]{8} & \num[text-series-to-math=true]{5} & \num[text-series-to-math=true]{13} & \num[text-series-to-math=true]{38} (\num[text-series-to-math=true]{0.37}) & \textbf{\num[text-series-to-math=true]{5}} & \textbf{\num[text-series-to-math=true]{5}} & \textbf{\num[text-series-to-math=true]{5}} & \textbf{\num[text-series-to-math=true]{5}} & 5.0 & 0.4\\
email & \num[text-series-to-math=true]{9} & \num[text-series-to-math=true]{7} & \num[text-series-to-math=true]{11} & \num[text-series-to-math=true]{70} (\num[text-series-to-math=true]{0.47}) & \textbf{\num[text-series-to-math=true]{6}} & \textbf{\num[text-series-to-math=true]{6}} & \textbf{\num[text-series-to-math=true]{6}} & \textbf{\num[text-series-to-math=true]{6}} & 5.1 & 0.4\\
car & \num[text-series-to-math=true]{14} & \num[text-series-to-math=true]{9} & \num[text-series-to-math=true]{26} & \num[text-series-to-math=true]{77} (\num[text-series-to-math=true]{0.06}) & \textbf{\num[text-series-to-math=true]{5}} & \textbf{\num[text-series-to-math=true]{5}} & \textbf{\num[text-series-to-math=true]{5}} & \textbf{\num[text-series-to-math=true]{5}} & 5.0 & 0.3\\
SafeBali & \num[text-series-to-math=true]{10} & \num[text-series-to-math=true]{10} & \num[text-series-to-math=true]{17} & \num[text-series-to-math=true]{141} (\num[text-series-to-math=true]{0.16}) & \textbf{\num[text-series-to-math=true]{11}} & \textbf{\num[text-series-to-math=true]{11}} & \textbf{\num[text-series-to-math=true]{11}} & \textbf{\num[text-series-to-math=true]{11}} & 5.0 & 0.3\\
ChatClient & \num[text-series-to-math=true]{13} & \num[text-series-to-math=true]{10} & \num[text-series-to-math=true]{18} & \num[text-series-to-math=true]{176} (\num[text-series-to-math=true]{0.64}) & \textbf{\num[text-series-to-math=true]{7}} & \textbf{\num[text-series-to-math=true]{7}} & \textbf{\num[text-series-to-math=true]{7}} & \textbf{\num[text-series-to-math=true]{7}} & 5.3 & 1.3\\
APL & \num[text-series-to-math=true]{15} & \num[text-series-to-math=true]{13} & \num[text-series-to-math=true]{19} & \num[text-series-to-math=true]{285} (\num[text-series-to-math=true]{0.58}) & \textbf{\num[text-series-to-math=true]{7}} & \textbf{\num[text-series-to-math=true]{7}} & \textbf{\num[text-series-to-math=true]{7}} & \textbf{\num[text-series-to-math=true]{7}} & 5.1 & 1.2\\
FameDB & \num[text-series-to-math=true]{15} & \num[text-series-to-math=true]{13} & \num[text-series-to-math=true]{26} & \num[text-series-to-math=true]{302} (\num[text-series-to-math=true]{0.39}) & \textbf{\num[text-series-to-math=true]{8}} & \textbf{\num[text-series-to-math=true]{8}} & \textbf{\num[text-series-to-math=true]{8}} & \textbf{\num[text-series-to-math=true]{8}} & 5.1 & 1.4\\
toybox\_2006-10\ldots & \num[text-series-to-math=true]{16} & \num[text-series-to-math=true]{16} & \num[text-series-to-math=true]{13} & \num[text-series-to-math=true]{457} (\num[text-series-to-math=true]{0.24}) & \textbf{\num[text-series-to-math=true]{8}} & \textbf{\num[text-series-to-math=true]{8}} & \textbf{\num[text-series-to-math=true]{8}} & \textbf{\num[text-series-to-math=true]{8}} & 5.1 & 1.6\\
FeatureIDE & \num[text-series-to-math=true]{21} & \num[text-series-to-math=true]{16} & \num[text-series-to-math=true]{29} & \num[text-series-to-math=true]{473} (\num[text-series-to-math=true]{0.45}) & \textbf{\num[text-series-to-math=true]{8}} & \textbf{\num[text-series-to-math=true]{8}} & \textbf{\num[text-series-to-math=true]{8}} & \textbf{\num[text-series-to-math=true]{8}} & 11.8 & 104.0\\
APL-Model & \num[text-series-to-math=true]{24} & \num[text-series-to-math=true]{21} & \num[text-series-to-math=true]{32} & \num[text-series-to-math=true]{782} (\num[text-series-to-math=true]{0.37}) & \textbf{\num[text-series-to-math=true]{8}} & \textbf{\num[text-series-to-math=true]{8}} & \textbf{\num[text-series-to-math=true]{8}} & \textbf{\num[text-series-to-math=true]{8}} & 6.6 & 13.9\\
TightVNC & \num[text-series-to-math=true]{23} & \num[text-series-to-math=true]{21} & \num[text-series-to-math=true]{29} & \num[text-series-to-math=true]{788} (\num[text-series-to-math=true]{0.48}) & \textbf{\num[text-series-to-math=true]{8}} & \textbf{\num[text-series-to-math=true]{8}} & \textbf{\num[text-series-to-math=true]{8}} & \textbf{\num[text-series-to-math=true]{8}} & 5.8 & 17.2\\
SortingLine & \num[text-series-to-math=true]{33} & \num[text-series-to-math=true]{23} & \num[text-series-to-math=true]{57} & \num[text-series-to-math=true]{948} (\num[text-series-to-math=true]{0.29}) & \textbf{\num[text-series-to-math=true]{9}} & \textbf{\num[text-series-to-math=true]{9}} & \textbf{\num[text-series-to-math=true]{9}} & \textbf{\num[text-series-to-math=true]{9}} & 5.1 & 7.9\\
gpl & \num[text-series-to-math=true]{39} & \num[text-series-to-math=true]{25} & \num[text-series-to-math=true]{88} & \num[text-series-to-math=true]{1067} (\num[text-series-to-math=true]{0.18}) & \textbf{\num[text-series-to-math=true]{16}} & \textbf{\num[text-series-to-math=true]{16}} & \textbf{\num[text-series-to-math=true]{16}} & \textbf{\num[text-series-to-math=true]{16}} & 5.1 & 4.4\\
PPU & \num[text-series-to-math=true]{29} & \num[text-series-to-math=true]{25} & \num[text-series-to-math=true]{58} & \num[text-series-to-math=true]{1112} (\num[text-series-to-math=true]{0.22}) & \textbf{\num[text-series-to-math=true]{12}} & \textbf{\num[text-series-to-math=true]{12}} & \textbf{\num[text-series-to-math=true]{12}} & \textbf{\num[text-series-to-math=true]{12}} & 5.0 & 4.3\\
dell & \num[text-series-to-math=true]{39} & \num[text-series-to-math=true]{37} & \num[text-series-to-math=true]{230} & \num[text-series-to-math=true]{2273} (\num[text-series-to-math=true]{0.07}) & \textbf{\num[text-series-to-math=true]{31}} & \textbf{\num[text-series-to-math=true]{31}} & \textbf{\num[text-series-to-math=true]{31}} & \textbf{\num[text-series-to-math=true]{31}} & 5.1 & 42.2\\
berkeleyDB1 & \num[text-series-to-math=true]{79} & \num[text-series-to-math=true]{53} & \num[text-series-to-math=true]{133} & \num[text-series-to-math=true]{5048} (\num[text-series-to-math=true]{0.21}) & \textbf{\num[text-series-to-math=true]{15}} & \textbf{\num[text-series-to-math=true]{15}} & \textbf{\num[text-series-to-math=true]{15}} & \textbf{\num[text-series-to-math=true]{15}} & 5.1 & 116.0\\
axTLS & \num[text-series-to-math=true]{106} & \num[text-series-to-math=true]{58} & \num[text-series-to-math=true]{177} & \num[text-series-to-math=true]{5596} (\num[text-series-to-math=true]{0.39}) & \num[text-series-to-math=true]{10} & \num[text-series-to-math=true]{11} & \num[text-series-to-math=true]{10} & \num[text-series-to-math=true]{11} & 3600.0 & 3604.7\\
Violet & \num[text-series-to-math=true]{115} & \num[text-series-to-math=true]{88} & \num[text-series-to-math=true]{181} & \num[text-series-to-math=true]{14517} (\num[text-series-to-math=true]{0.20}) & \num[text-series-to-math=true]{16}, \textbf{\num[text-series-to-math=true]{17}}, \textbf{\num[text-series-to-math=true]{17}} & \textbf{\num[text-series-to-math=true]{17}} & \num[text-series-to-math=true]{16} & \textbf{\num[text-series-to-math=true]{17}} & 958.0 & 3621.4\\
berkeleyDB2 & \num[text-series-to-math=true]{233} & \num[text-series-to-math=true]{103} & \num[text-series-to-math=true]{507} & \num[text-series-to-math=true]{20061} (\num[text-series-to-math=true]{0.06}) & \textbf{\num[text-series-to-math=true]{12}} & \textbf{\num[text-series-to-math=true]{12}} & \textbf{\num[text-series-to-math=true]{12}} & \textbf{\num[text-series-to-math=true]{12}} & 5.2 & 185.2\\
soletta\_2015-0\ldots & \num[text-series-to-math=true]{129} & \num[text-series-to-math=true]{129} & \num[text-series-to-math=true]{192} & \num[text-series-to-math=true]{24018} (\num[text-series-to-math=true]{0.25}) & \textbf{\num[text-series-to-math=true]{24}} & \textbf{\num[text-series-to-math=true]{24}} & \textbf{\num[text-series-to-math=true]{24}} & \textbf{\num[text-series-to-math=true]{24}} & 5.2 & 23.4\\
BankingSoftware & \num[text-series-to-math=true]{128} & \num[text-series-to-math=true]{122} & \num[text-series-to-math=true]{184} & \num[text-series-to-math=true]{29188} (\num[text-series-to-math=true]{0.55}) & \textbf{\num[text-series-to-math=true]{29}} & \textbf{\num[text-series-to-math=true]{29}} & \textbf{\num[text-series-to-math=true]{29}} & \textbf{\num[text-series-to-math=true]{29}} & 5.1 & 281.5\\
BattleofTanks & \num[text-series-to-math=true]{137} & \num[text-series-to-math=true]{131} & \num[text-series-to-math=true]{755} & \num[text-series-to-math=true]{33452} (\num[text-series-to-math=true]{0.38}) & \num[text-series-to-math=true]{256} & \num[text-series-to-math=true]{283}, \num[text-series-to-math=true]{290}, \num[text-series-to-math=true]{301} & \num[text-series-to-math=true]{256} & \num[text-series-to-math=true]{294}, \num[text-series-to-math=true]{296}, \num[text-series-to-math=true]{298} & 3600.0 & 3613.6\\
E-Shop & \num[text-series-to-math=true]{213} & \num[text-series-to-math=true]{175} & \num[text-series-to-math=true]{273} & \num[text-series-to-math=true]{60465} (\num[text-series-to-math=true]{0.59}) & \num[text-series-to-math=true]{11} & \num[text-series-to-math=true]{12} & \num[text-series-to-math=true]{9}, \num[text-series-to-math=true]{10}, \num[text-series-to-math=true]{10} & \num[text-series-to-math=true]{12} & 3600.0 & 3631.6\\
fiasco\_2017-09\ldots & \num[text-series-to-math=true]{230} & \num[text-series-to-math=true]{230} & \num[text-series-to-math=true]{1181} & \num[text-series-to-math=true]{83612} (\num[text-series-to-math=true]{0.06}) & \textbf{\num[text-series-to-math=true]{225}} & \textbf{\num[text-series-to-math=true]{225}} & \textbf{\num[text-series-to-math=true]{225}} & \textbf{\num[text-series-to-math=true]{225}} & 6.2 & 777.1\\
fiasco\_2020-12\ldots & \num[text-series-to-math=true]{258} & \num[text-series-to-math=true]{258} & \num[text-series-to-math=true]{1542} & \num[text-series-to-math=true]{102755} (\num[text-series-to-math=true]{0.05}) & \textbf{\num[text-series-to-math=true]{196}} & \textbf{\num[text-series-to-math=true]{196}} & \textbf{\num[text-series-to-math=true]{196}} & \textbf{\num[text-series-to-math=true]{196}} & 6.2 & 323.7\\
uclibc\_2008-06\ldots & \num[text-series-to-math=true]{263} & \num[text-series-to-math=true]{263} & \num[text-series-to-math=true]{1699} & \num[text-series-to-math=true]{127211} (\num[text-series-to-math=true]{0.21}) & \textbf{\num[text-series-to-math=true]{505}} & \textbf{\num[text-series-to-math=true]{505}} & \textbf{\num[text-series-to-math=true]{505}} & \textbf{\num[text-series-to-math=true]{505}} & 5.2 & 120.8\\
uclibc\_2020-12\ldots & \num[text-series-to-math=true]{272} & \num[text-series-to-math=true]{272} & \num[text-series-to-math=true]{1670} & \num[text-series-to-math=true]{135694} (\num[text-series-to-math=true]{0.19}) & \textbf{\num[text-series-to-math=true]{365}} & \textbf{\num[text-series-to-math=true]{365}} & \textbf{\num[text-series-to-math=true]{365}} & \textbf{\num[text-series-to-math=true]{365}} & 5.2 & 44.8\\
toybox\_2020-12\ldots & \num[text-series-to-math=true]{334} & \num[text-series-to-math=true]{334} & \num[text-series-to-math=true]{92} & \num[text-series-to-math=true]{206665} (\num[text-series-to-math=true]{0.75}) & \num[text-series-to-math=true]{8} & \num[text-series-to-math=true]{13} & \num[text-series-to-math=true]{7}, \num[text-series-to-math=true]{8}, \num[text-series-to-math=true]{8} & \num[text-series-to-math=true]{13} & 3600.0 & 3614.5\\
DMIE & \num[text-series-to-math=true]{346} & \num[text-series-to-math=true]{344} & \num[text-series-to-math=true]{587} & \num[text-series-to-math=true]{235264} (\num[text-series-to-math=true]{0.91}) & \textbf{\num[text-series-to-math=true]{16}} & \textbf{\num[text-series-to-math=true]{16}} & \textbf{\num[text-series-to-math=true]{16}} & \textbf{\num[text-series-to-math=true]{16}} & 6.3 & 101.8\\
soletta\_2017-0\ldots & \num[text-series-to-math=true]{458} & \num[text-series-to-math=true]{458} & \num[text-series-to-math=true]{1862} & \num[text-series-to-math=true]{269353} (\num[text-series-to-math=true]{0.08}) & \textbf{\num[text-series-to-math=true]{37}} & \textbf{\num[text-series-to-math=true]{37}} & \num[text-series-to-math=true]{31}, \num[text-series-to-math=true]{34}, \textbf{\num[text-series-to-math=true]{37}} & \textbf{\num[text-series-to-math=true]{37}} & 7.0 & 3609.1\\
fs\_2017-05-22 & \num[text-series-to-math=true]{1342} & \num[text-series-to-math=true]{430} & \num[text-series-to-math=true]{5714} & \num[text-series-to-math=true]{277429} (\num[text-series-to-math=true]{0.00}) & \textbf{\num[text-series-to-math=true]{396}} & \textbf{\num[text-series-to-math=true]{396}} & \textbf{\num[text-series-to-math=true]{396}} & \textbf{\num[text-series-to-math=true]{396}} & 5.3 & 371.2\\
WaterlooGenerated & \num[text-series-to-math=true]{580} & \num[text-series-to-math=true]{423} & \num[text-series-to-math=true]{877} & \num[text-series-to-math=true]{347472} (\num[text-series-to-math=true]{0.54}) & \textbf{\num[text-series-to-math=true]{82}} & \textbf{\num[text-series-to-math=true]{82}} & \textbf{\num[text-series-to-math=true]{82}} & \textbf{\num[text-series-to-math=true]{82}} & 5.5 & 147.4\\
financial\_serv\ldots & \num[text-series-to-math=true]{2129} & \num[text-series-to-math=true]{585} & \num[text-series-to-math=true]{7898} & \num[text-series-to-math=true]{537635} (\num[text-series-to-math=true]{0.03}) & \textbf{\num[text-series-to-math=true]{4340}} & \textbf{\num[text-series-to-math=true]{4340}} & \num[text-series-to-math=true]{4324}, \num[text-series-to-math=true]{4336}, \num[text-series-to-math=true]{4336} & \textbf{\num[text-series-to-math=true]{4340}}, \textbf{\num[text-series-to-math=true]{4340}}, \num[text-series-to-math=true]{4341} & 135.1 & 3652.6\\
financial-servi\ldots & \num[text-series-to-math=true]{2126} & \num[text-series-to-math=true]{589} & \num[text-series-to-math=true]{7839} & \num[text-series-to-math=true]{542596} (\num[text-series-to-math=true]{0.03}) & \textbf{\num[text-series-to-math=true]{4352}} & \textbf{\num[text-series-to-math=true]{4352}} & \num[text-series-to-math=true]{4334}, \num[text-series-to-math=true]{4336}, \num[text-series-to-math=true]{4348} & \textbf{\num[text-series-to-math=true]{4352}}, \textbf{\num[text-series-to-math=true]{4352}}, \num[text-series-to-math=true]{4355} & 119.7 & 3630.4\\
busybox\_2007-0\ldots & \num[text-series-to-math=true]{540} & \num[text-series-to-math=true]{540} & \num[text-series-to-math=true]{429} & \num[text-series-to-math=true]{575209} (\num[text-series-to-math=true]{0.38}) & \textbf{\num[text-series-to-math=true]{21}} & \textbf{\num[text-series-to-math=true]{21}} & \textbf{\num[text-series-to-math=true]{21}} & \textbf{\num[text-series-to-math=true]{21}} & 6.3 & 134.1\\
busybox-1\_18\_0 & \num[text-series-to-math=true]{1186} & \num[text-series-to-math=true]{667} & \num[text-series-to-math=true]{1596} & \num[text-series-to-math=true]{865712} (\num[text-series-to-math=true]{0.67}) & \num[text-series-to-math=true]{13} & \num[text-series-to-math=true]{16} & \num[text-series-to-math=true]{12}, \num[text-series-to-math=true]{12}, \num[text-series-to-math=true]{13} & \num[text-series-to-math=true]{17} & 3600.5 & 3621.6\\
am31\_sim & \num[text-series-to-math=true]{1017} & \num[text-series-to-math=true]{725} & \num[text-series-to-math=true]{1865} & \num[text-series-to-math=true]{975956} (\num[text-series-to-math=true]{0.10}) & \num[text-series-to-math=true]{29} & \num[text-series-to-math=true]{33} & \num[text-series-to-math=true]{24}, \num[text-series-to-math=true]{28}, \num[text-series-to-math=true]{29} & \num[text-series-to-math=true]{35}, \num[text-series-to-math=true]{35}, \num[text-series-to-math=true]{36} & 3600.2 & 3627.8\\
atlas\_mips32\_\ldots & \num[text-series-to-math=true]{1066} & \num[text-series-to-math=true]{753} & \num[text-series-to-math=true]{1955} & \num[text-series-to-math=true]{1061278} (\num[text-series-to-math=true]{0.10}) & \textbf{\num[text-series-to-math=true]{34}} & \textbf{\num[text-series-to-math=true]{34}} & \num[text-series-to-math=true]{31}, \num[text-series-to-math=true]{33}, \num[text-series-to-math=true]{33} & \num[text-series-to-math=true]{37}, \num[text-series-to-math=true]{37}, \num[text-series-to-math=true]{38} & 811.1 & 3633.8\\
EMBToolkit & \num[text-series-to-math=true]{7990} & \num[text-series-to-math=true]{850} & \num[text-series-to-math=true]{24004} & \num[text-series-to-math=true]{1074596} (\num[text-series-to-math=true]{0.20}) & \textbf{\num[text-series-to-math=true]{1872}} & \textbf{\num[text-series-to-math=true]{1872}} & \textbf{\num[text-series-to-math=true]{1872}} & \textbf{\num[text-series-to-math=true]{1872}} & 5.8 & 1365.8\\
integrator\_arm7 & \num[text-series-to-math=true]{1162} & \num[text-series-to-math=true]{794} & \num[text-series-to-math=true]{2138} & \num[text-series-to-math=true]{1174231} (\num[text-series-to-math=true]{0.10}) & \textbf{\num[text-series-to-math=true]{34}} & \textbf{\num[text-series-to-math=true]{34}} & \num[text-series-to-math=true]{27}, \num[text-series-to-math=true]{32}, \num[text-series-to-math=true]{33} & \num[text-series-to-math=true]{37}, \num[text-series-to-math=true]{37}, \num[text-series-to-math=true]{38} & 784.2 & 3619.2\\
XSEngine & \num[text-series-to-math=true]{1158} & \num[text-series-to-math=true]{804} & \num[text-series-to-math=true]{2096} & \num[text-series-to-math=true]{1210521} (\num[text-series-to-math=true]{0.10}) & \textbf{\num[text-series-to-math=true]{34}} & \textbf{\num[text-series-to-math=true]{34}} & \num[text-series-to-math=true]{28}, \num[text-series-to-math=true]{30}, \num[text-series-to-math=true]{32} & \num[text-series-to-math=true]{37}, \num[text-series-to-math=true]{38}, \num[text-series-to-math=true]{38} & 407.9 & 3612.2\\
aaed2000 & \num[text-series-to-math=true]{1183} & \num[text-series-to-math=true]{819} & \num[text-series-to-math=true]{2190} & \num[text-series-to-math=true]{1250781} (\num[text-series-to-math=true]{0.10}) & \textbf{\num[text-series-to-math=true]{52}} & \textbf{\num[text-series-to-math=true]{52}} & \num[text-series-to-math=true]{46}, \num[text-series-to-math=true]{51}, \num[text-series-to-math=true]{51} & \textbf{\num[text-series-to-math=true]{52}}, \num[text-series-to-math=true]{53}, \num[text-series-to-math=true]{53} & 260.7 & 3611.8\\
ea2468 & \num[text-series-to-math=true]{1226} & \num[text-series-to-math=true]{849} & \num[text-series-to-math=true]{2235} & \num[text-series-to-math=true]{1286744} (\num[text-series-to-math=true]{0.09}) & \textbf{\num[text-series-to-math=true]{34}} & \textbf{\num[text-series-to-math=true]{34}} & \num[text-series-to-math=true]{29}, \num[text-series-to-math=true]{31}, \num[text-series-to-math=true]{32} & \num[text-series-to-math=true]{37}, \num[text-series-to-math=true]{38}, \num[text-series-to-math=true]{39} & 1739.3 & 3611.2\\
busybox-1\_29\_2 & \num[text-series-to-math=true]{1018} & \num[text-series-to-math=true]{1018} & \num[text-series-to-math=true]{997} & \num[text-series-to-math=true]{2045110} (\num[text-series-to-math=true]{0.34}) & \textbf{\num[text-series-to-math=true]{22}} & \textbf{\num[text-series-to-math=true]{22}} & \num[text-series-to-math=true]{18}, \num[text-series-to-math=true]{18}, \num[text-series-to-math=true]{19} & \textbf{\num[text-series-to-math=true]{22}} & 125.1 & 3617.0\\
busybox\_2020-1\ldots & \num[text-series-to-math=true]{1050} & \num[text-series-to-math=true]{1050} & \num[text-series-to-math=true]{996} & \num[text-series-to-math=true]{2176803} (\num[text-series-to-math=true]{0.37}) & \textbf{\num[text-series-to-math=true]{20}} & \textbf{\num[text-series-to-math=true]{20}} & \num[text-series-to-math=true]{18}, \num[text-series-to-math=true]{18}, \num[text-series-to-math=true]{19} & \num[text-series-to-math=true]{21} & 30.4 & 3637.4\\
eCos-3-0\_i386pc & \num[text-series-to-math=true]{1245} & \num[text-series-to-math=true]{1244} & \num[text-series-to-math=true]{3723} & \num[text-series-to-math=true]{2910229} (\num[text-series-to-math=true]{0.04}) & \textbf{\num[text-series-to-math=true]{37}} & \textbf{\num[text-series-to-math=true]{37}} & \num[text-series-to-math=true]{23}, \num[text-series-to-math=true]{31}, \num[text-series-to-math=true]{34} & \num[text-series-to-math=true]{41}, \num[text-series-to-math=true]{42}, \num[text-series-to-math=true]{42} & 715.6 & 3641.6\\
FreeBSD-8\_0\_0 & \num[text-series-to-math=true]{1397} & \num[text-series-to-math=true]{1396} & \num[text-series-to-math=true]{15692} & \num[text-series-to-math=true]{3765597} (\num[text-series-to-math=true]{0.35}) & \num[text-series-to-math=true]{33}, \num[text-series-to-math=true]{34}, \num[text-series-to-math=true]{34} & \num[text-series-to-math=true]{40} & \num[text-series-to-math=true]{28}, \num[text-series-to-math=true]{29}, \num[text-series-to-math=true]{30} & \num[text-series-to-math=true]{43}, \num[text-series-to-math=true]{44}, \num[text-series-to-math=true]{45} & 3602.9 & 3636.0\\
Automotive01 & \num[text-series-to-math=true]{2047} & \num[text-series-to-math=true]{1790} & \num[text-series-to-math=true]{9117} & \num[text-series-to-math=true]{5772017} (\num[text-series-to-math=true]{0.12}) & \num[text-series-to-math=true]{529}, \num[text-series-to-math=true]{530}, \num[text-series-to-math=true]{530} & \num[text-series-to-math=true]{727}, \num[text-series-to-math=true]{728}, \num[text-series-to-math=true]{734} & \num[text-series-to-math=true]{266}, \num[text-series-to-math=true]{417}, \num[text-series-to-math=true]{525} & \num[text-series-to-math=true]{795}, \num[text-series-to-math=true]{806}, \num[text-series-to-math=true]{814} & 3600.1 & 3637.8\\
linux\_2\_6\_33\ldots & \num[text-series-to-math=true]{15404} & \num[text-series-to-math=true]{5432} & \num[text-series-to-math=true]{33729} & \num[text-series-to-math=true]{55249944} (\num[text-series-to-math=true]{0.37}) & \textbf{\num[text-series-to-math=true]{483}} & \textbf{\num[text-series-to-math=true]{483}} & \num[text-series-to-math=true]{83}, \num[text-series-to-math=true]{261}, \num[text-series-to-math=true]{317} & \textbf{\num[text-series-to-math=true]{483}} & 856.5 & 3816.8\\
linux\_2\_6\_28\ldots & \num[text-series-to-math=true]{6889} & \num[text-series-to-math=true]{6888} & \num[text-series-to-math=true]{87604} & \num[text-series-to-math=true]{92540449} (\num[text-series-to-math=true]{0.11}) & \textbf{\num[text-series-to-math=true]{433}} & \textbf{\num[text-series-to-math=true]{433}} & \num[text-series-to-math=true]{144}, \num[text-series-to-math=true]{289}, \num[text-series-to-math=true]{429} & \num[text-series-to-math=true]{486}, \num[text-series-to-math=true]{490}, \num[text-series-to-math=true]{508} & 328.1 & 3955.7\\
Automotive02\_V1 & \num[text-series-to-math=true]{12478} & \num[text-series-to-math=true]{12289} & \num[text-series-to-math=true]{234472} & \num[text-series-to-math=true]{301619972} (\num[text-series-to-math=true]{0.32}) & \textbf{\num[text-series-to-math=true]{37327}} & \textbf{\num[text-series-to-math=true]{37327}} & --- & --- & 1823.0 & ---\\
Automotive02\_V2 & \num[text-series-to-math=true]{15883} & \num[text-series-to-math=true]{15683} & \num[text-series-to-math=true]{339036} & \num[text-series-to-math=true]{491275076} (\num[text-series-to-math=true]{0.30}) & \textbf{\num[text-series-to-math=true]{38115}} & \textbf{\num[text-series-to-math=true]{38115}} & --- & --- & 2101.9 & ---\\
Automotive02\_V3 & \num[text-series-to-math=true]{16607} & \num[text-series-to-math=true]{16224} & \num[text-series-to-math=true]{343538} & \num[text-series-to-math=true]{525717328} (\num[text-series-to-math=true]{0.30}) & \textbf{\num[text-series-to-math=true]{38115}} & \textbf{\num[text-series-to-math=true]{38115}} & --- & --- & 2213.3 & ---\\
Automotive02\_V4 & \num[text-series-to-math=true]{16809} & \num[text-series-to-math=true]{16376} & \num[text-series-to-math=true]{346172} & \num[text-series-to-math=true]{535654597} (\num[text-series-to-math=true]{0.30}) & \textbf{\num[text-series-to-math=true]{38577}} & \textbf{\num[text-series-to-math=true]{38577}} & --- & --- & 2272.8 & ---\\
\hline\ \\
\end{tabular}}}
\end{tiny}
\vspace{-0.2cm}
\caption{\begin{small}Table of the instances of the benchmark set and the outcomes of Sammy and SampLNS.
Bold numbers indicate optimal solutions.
The bound columns give minimum, median and maximum values achieved by the \num{5} repeat runs we performed per instance and algorithm, unless all three numbers are the same.
The runtime shown is the median runtime across repeat runs.
For Sammy, we ran the initial heuristic for a minimum of \qty{5}{s}, therefore the runtime, even for very small instances,
is never below \qty{5}{s} since the initial lower bound was never sufficient to prove optimality.
The number of feasible interactions given is before simplification and reduction;
the fraction remaining after simplification and reduction is given in parentheses.\end{small}}
\label{tab:instance_summary}
\end{table*}

We also address (at least in part) additional research questions regarding the LNS repair subproblems,
and assess the performance of our implementation on a large superset of our benchmark instance set comprised of \num{1148} different instances.
\begin{description}
  \item[RQ5] How do the different SAT solvers and repair approaches perform relative to each other on difficult instances?
  \item[RQ6] How is the size of the gap between our symmetry breaker and the number of allowed configurations distributed?
  \item[RQ7] How does the quality of the symmetry breaker affect the performance of our repair approaches?
\end{description}
To assess these questions, we first run Sammy with a time limit of \qty{1}{h} on each of the \num{1148} instances.
We then eliminate the instances that are trivial or at least relatively easy to solve to optimality,
and only retain the instances that were not solved to provable optimality within \qty{10}{min}.
Not including the large \texttt{AutomotiveV02} instances, which are solved relatively quickly for their size but take considerable time in the initial phase,
\num{122} instances (\qty{10.7}{\percent} of the \num{1144} considered instances) remain;
in other words, \qty{89.3}{\percent} of the instances in the large instance set could be solved to provable optimality within \qty{10}{min}.

On the remaining instances, we reran Sammy, exporting each LNS repair subproblem produced by our destroy operations,
resulting in a total of \num{103455} exported subproblems.
We then ran our mutually exclusive set heuristic on each subproblem for a default \num{10} cutting plane or pricing iterations.
In other words, on each subproblem, we completed the inner loop eliminating all violated non-edges from the relaxation \num{10} times, 
adding cutting planes or additional interactions via pricing after each completion unless the mutually exclusive set was found to be optimal.
This number of iterations is also the default that was used during the other experiments and was identified 
as sensible trade-off between runtime and symmetry breaker quality by preliminary experiments.
Though usually much quicker at a median runtime of just \qty{0.07}{s} per subproblem, this took at most \qty{1}{s} per subproblem 
(only measuring time actually spent computing mutually exclusive sets).

For a total of \num{16634} (\qty{16.1}{\percent}) of the subproblems, that proved the existing coverage to be optimal by finding a matching mutually exclusive set;
the full distribution of the gaps between the number of allowed configurations and the symmetry-breaking mutually exclusive set is shown in \cref{fig:gap-to-allowed-nconf}.
To answer RQ6, we see that the majority of subproblems have a gap of at most $1$ between the number of allowed configurations and the symmetry breaker,
and over \qty{90}{\percent} have a gap of at most $3$.
\begin{figure}%
  \centering%
  \includegraphics{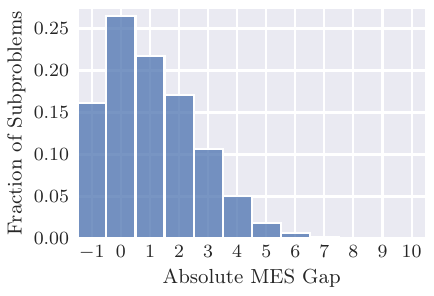}%
  \caption{Histogram showing the absolute gap values between the
           number of allowed configurations in the repair subproblem and the 
           size of the mutually exclusive set used as symmetry breaker.
           A value of $-1$ means that the existing solution is optimal.}
  \label{fig:gap-to-allowed-nconf}
\end{figure}

Excluding the subproblems shown to be optimally solved by the symmetry breaker, \num{86821} subproblems remained.
From these remaining subproblems, we chose \num{500} subproblems uniformly at random and ran each of the \num{13} possible approach/solver combinations on each with a time limit of \qty{30}{min}.
Each approach was given the same stored mutually exclusive set; only the alternating LB-UB approach attempts to improve upon that set during the subproblem solve.
The performance of the individual approaches is shown in \cref{fig:subproblem-solver-comparison}.
In total, \num{478} of the \num{500} subproblems were solved within the time limit by at least one approach/solver combination.
We see that, in general, the individual approaches are relatively close in terms of performance;
looking more precisely, the alternating LB-UB approach appears to perform slightly worse than the others.
Similarly, using Lingeling and Cryptominisat as backend seems to perform slightly worse than CaDiCaL or kissat (which does not support incremental solving).
The best approaches seem to be non-incremental kissat as well as either of the non-alternating incremental approaches based on CaDiCaL.
However, in particular when considering lower runtimes, the virtual best solver, 
i.e., assuming that an oracle told us in advance which solver to use, has a recognizable lead on any individual solver.
The outcome (i.e., whether an improved assignment is possible or not) seems to be important when it comes to the relative performance.
For subproblems where an improvement is eventually found, the non-incremental approach based on kissat is almost as good as the virtual best solver;
for infeasible subproblems, the greedy incremental approach based on CaDiCaL appears to be better suited, in particular when considering lower runtimes,
with the simple incremental approach and CaDiCaL in between in either case.
To partially answer RQ5, results suggest that among the evaluated SAT solvers, CaDiCaL and kissat appear to be superior by a small margin,
and that the alternating LB-UB-approach, which is still useful as full problem approach that can make use of improved lower bounds as it (or the LB worker) find them,
might not be worthwhile as a repair subproblem solver.
They also suggest that the non-incremental approach may be slightly superior when it comes to finding improved assignments,
whereas the incremental approaches may be better at proving infeasibility of repair subproblems.
\begin{figure*}
  \centering
  \includegraphics{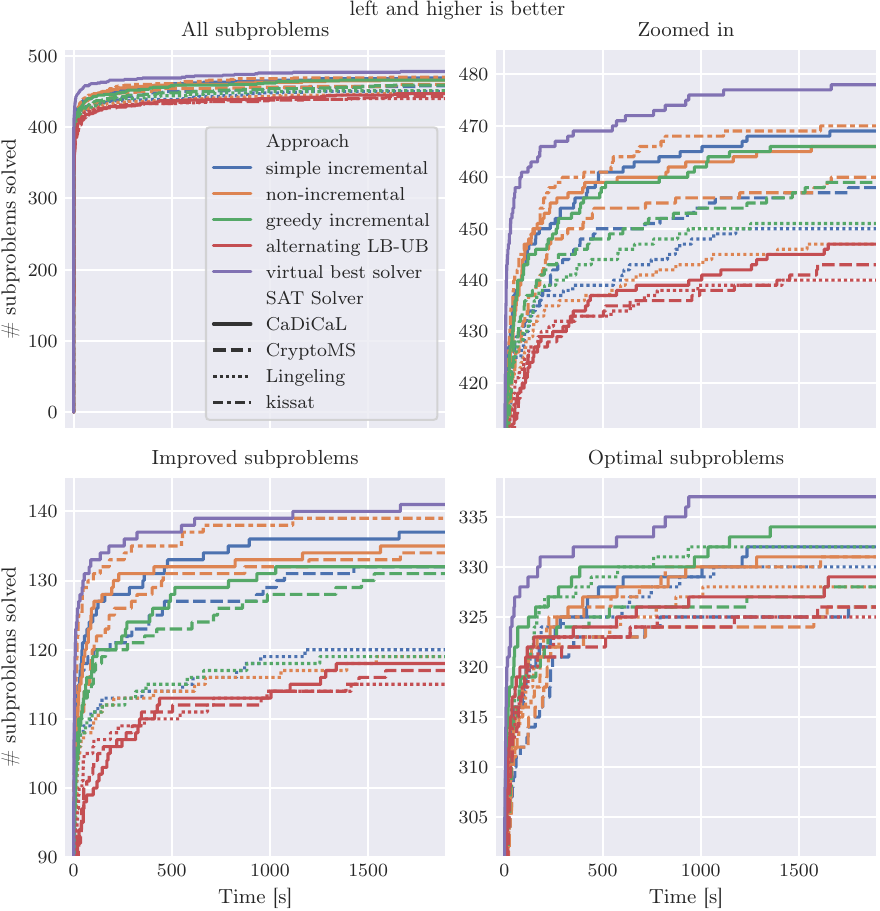}%
  \caption{The performance (repair subproblems solved over time) of the individual approach/solver combinations.
           The top row shows all subproblems, the bottom row shows only the subproblems that yielded improved solutions (left)
           and those for which the existing solution was already optimal (left).}
  \label{fig:subproblem-solver-comparison}
\end{figure*}

Finally, to address RQ7, consider \cref{fig:subproblem-solvers-by-gap}, which shows the performance of the individual solvers
considering subproblems for which the gap between symmetry breaker and allowed configurations falls into a certain range.
We see that subproblems with a low gap are almost all solved very quickly and with little performance difference between the individual approaches.
The performance degrades and becomes less homogeneous with growing gap.
This may suggest that, as one might expect, the impact of our symmetry breaker is substantial;
however, a good part of the performance difference may also be reflective of the fact that we find better symmetry breaker for easier subproblems.
To resolve this, we might consider spending more time on finding better mutually exclusive sets for the subproblems with nontrivial gaps before re-running them.
It may however also be advisable to tune the time spent searching for mutually exclusive sets based on the absolute gap; we leave this as future work.
\begin{figure*}%
  \centering%
  \includegraphics{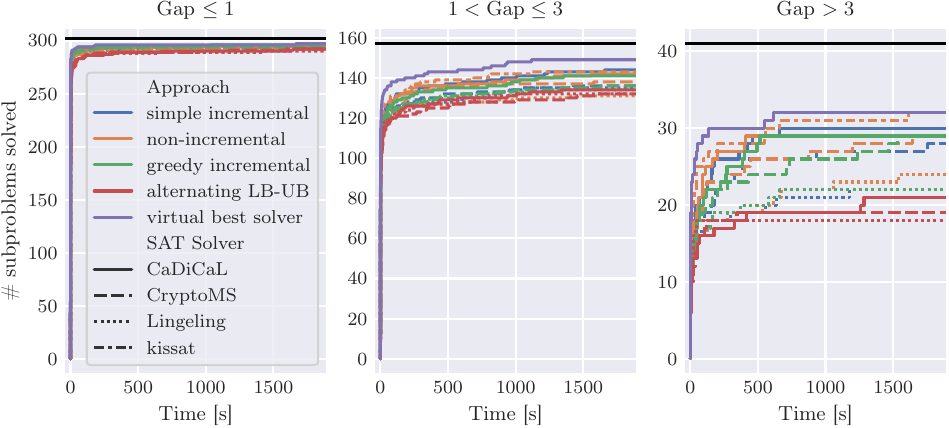}%
  \caption{The performance of the subproblem solvers for instances with very low gap between symmetry breaker and allowed configurations (left),
           moderate gap (center) and relatively large gap (right).
           The black line indicates the total number of subproblem in that gap range.}
  \label{fig:subproblem-solvers-by-gap}
\end{figure*}



\clearpage
\bibliographystyle{siamplain}
\bibliography{references}
\end{document}